\documentclass[journal,transmag]{IEEEtran}
\usepackage[utf8]{inputenc}
\usepackage{amsmath}
\usepackage{amssymb}
\usepackage[square,numbers,comma]{natbib}
\usepackage{graphicx}
\usepackage{glossaries-prefix}
\newacronym{pid}{PID}{proportional-integral-derivative}
\newacronym[prefixfirst={a\ },
            prefix={an\ }
            ]{ml}{ML}{machine learning}
\newacronym{drl}{DRL}{deep reinforcement learning}
\newacronym{mlp}{MLP}{multi layer perceptron}
\newacronym[prefixfirst={a\ },
            prefix={an\ }
            ]{nn}{NN}{neural network}
\newacronym[prefixfirst={a\ },
            prefix={an\ }
            ]{rl}{RL}{reinforcement learning}
\newacronym[prefixfirst={a\ },
            prefix={an\ }
            ]{mdp}{MDP}{Markov decision process}
\newacronym[prefixfirst={a\ },
            prefix={an\ }
            ]{mpc}{MPC}{model predictive control}
\newacronym{nmpc}{NMPC}{nonlinear model predictive control}
\newacronym{empc}{EMPC}{economic model predictive control}
\newacronym[prefixfirst={a\ },
            prefix={an\ }
            ]{mpcr}{MPC}{model predictive controller}
\newacronym[prefixfirst={a\ },
            prefix={an\ }
            ]{lqr}{LQR}{linear quadratic regulator}
\newacronym[prefixfirst={a\ },
            prefix={an\ }
            ]{lqg}{LQG}{linear qudratic gaussian}
\newacronym{fpga}{FPGA}{field-programmable gate array}
\newacronym{soc}{SoC}{state of charge}
\newacronym{ocp}{OCP}{optimal control problem}
\newacronym{ou}{OU}{Ornstein–Uhlenbeck}
\newacronym{uav}{UAV}{unmanned aerial vehicle}
\newacronym{ppo}{PPO}{proximal policy optimization}
\newacronym{trpo}{TRPO}{trust region policy optimization}
\newacronym{gpd}{GPD}{generalized Poisson distribution}
\newacronym{pmf}{PMF}{probability mass function}
\newacronym{cdf}{CDF}{cumulative distribution function}
\newacronym{gdp}{GDP}{generalized Poisson distribution}
\newacronym{dare}{DARE}{discrete-time algebraic Riccati equation}
\newacronym{ift}{IFT}{implicit function theorem}
\newacronym{flops}{FLOPS}{floating point operations per second}
\newacronym{gae}{GAE}{generalized advantage estimation}
\newacronym[prefixfirst={a\ },
            prefix={an\ }
            ]{sdp}{SDP}{semidefinite program}
\newacronym{ahmpc}{AHMPC}{adaptive horizon model predictive control}
\newacronym{msbe}{MSBE}{mean squared Bellman error}
\newacronym{spsa}{SPSA}{simultaneous perturbation stochastic approximation}

\usepackage[caption=false]{subfig}
\usepackage[algo2e, vlined, ruled, noend]{algorithm2e}
\usepackage{booktabs}
\newcommand{\ra}[1]{\renewcommand{\arraystretch}{#1}}
\usepackage{amsthm}

\usepackage[colorinlistoftodos]{todonotes}

\usepackage{bbm}  

\usepackage{mathtools}

\DeclarePairedDelimiter\floor{\lfloor}{\rfloor}

\newcommand{\matr}[1]{\mathbf{#1}}
\newcommand{\param}[1]{\mathrm{#1}}

\newtheorem{proposition}{Proposition}
\newtheorem{assumption}{Assumption}

\title{Optimization of the Model Predictive Control Meta-Parameters Through Reinforcement Learning}
\author{\IEEEauthorblockN{Eivind Bøhn\IEEEauthorrefmark{1},
Sebastien Gros\IEEEauthorrefmark{2},
Signe Moe\IEEEauthorrefmark{3}, and
Tor Arne Johansen\IEEEauthorrefmark{2,4}}
\IEEEauthorblockA{\IEEEauthorrefmark{1}Artificial Intelligence and Data Analytics, SINTEF Digital, Oslo, Norway}
\IEEEauthorblockA{\IEEEauthorrefmark{2}Department of Engineering Cybernetics, NTNU, Trondheim, Norway}
\IEEEauthorblockA{\IEEEauthorrefmark{3}Sopra Steria Applications, Oslo, Norway}
\IEEEauthorblockA{\IEEEauthorrefmark{4}Centre for Autonomous Marine Operations and Systems, NTNU, Trondheim, Norway}
\thanks{
This work has been submitted to the IEEE for possible publication. Copyright may be transferred without notice, after which this version may no longer be accessible. Funded by the Research Council of Norway through the Centres of Excellence funding scheme, grant number 223254 NTNU AMOS, and through PhD Scholarships at SINTEF, grant number 272402.}}

\begin{document}

\maketitle
\begin{abstract}
    \Gls{mpc} is increasingly being considered for control of fast systems and embedded applications. However, the \gls{mpc} has some significant challenges for such systems. Its high computational complexity results in high power consumption from the control algorithm, which could account for a significant share of the energy resources in battery-powered embedded systems. The MPC parameters must be tuned, which is largely a trial-and-error process that affects the control performance, the robustness and the computational complexity of the controller to a high degree. In this paper we propose a novel framework in which any parameter of the control algorithm can be jointly tuned using \gls{rl}, with the goal of simultaneously optimizing the control performance and the power usage of the control algorithm. We propose the novel idea of optimizing the meta-parameters of \gls{mpc} with \gls{rl}, i.e. parameters affecting the structure of the \gls{mpc} problem as opposed to the solution to a given problem. Our control algorithm is based on an event-triggered \gls{mpc} where we learn when the \gls{mpc} should be re-computed, and a dual mode \gls{mpc} and linear state feedback control law applied in between \gls{mpc} computations. We formulate a novel mixture-distribution policy, and show that with joint optimization we achieve improvements that do not present themselves when optimizing the same parameters in isolation. We demonstrate our framework on the inverted pendulum control task, reducing the total computation time of the control system by 36\% while also improving the control performance by 18.4\% over the best-performing \gls{mpc} baseline. 
\end{abstract}
\glsresetall

\section{Introduction}
\Gls{mpc} is a powerful optimizing control technique, capable of controlling a wide range of systems with high control proficiency while respecting system constraints. Nonlinear \gls{mpc} can even handle nonlinear dynamics and constraints, and while not as simple as its linear counterpart, is increasingly being considered for embedded systems applications with fast dynamics \cite{gros2012aircraft,albin2015nonlinear,Joh17}. However, one of the main drawbacks of the \gls{mpc} is its high computational complexity, which makes it ill-suited for applications with low-powered hardware or battery energy restrictions, necessitating some form of compromise in its implementation \cite{le_approximate_2012, gondhalekar_tackling_2015}.

The high computational complexity of the \gls{mpc} comes as a result of its online operation consisting of solving a numerical \gls{ocp} at every time step, executing the first control input of the computed optimal solution, and then solving the \gls{ocp} again at the subsequent time step. Several adjustments to this basic paradigm has been suggested in order to reduce the computational complexity, e.g. early termination (suboptimal) \gls{mpc} \citep{scokaert_suboptimal_1999}, semi-explicit \gls{mpc} \citep{goebel_simple_2015}, and explicit \gls{mpc} \citep{bemporad_explicit_2002}.

A further challenge of the \gls{mpc} method is the need to tune its parameters to the task at hand, greatly affecting the control proficiency, robustness and computational complexity of the controller. The main tunable parameter in these regards is the prediction horizon,  which essentially controls how far into the future the \gls{mpc} evaluates the optimality of its solution. Further, the performance sensitivity to the prediction horizon varies over the state space, and this observation motivated the adaptive-horizon \gls{mpc} technique, in which the prediction horizon is varied according to some criteria. There exists several suggestions in the literature on how to design these criteria, e.g. terminal conditions \cite{michalska_robust_1993,krener_adaptive_2018}, as decision variables of the \gls{ocp} \cite{scokaert_min_1998}, and learning approaches \cite{gardezi_machine_2018,bohn2021reinforcement}. Other parameters of the \gls{mpc} scheme are also subject to tuning, e.g. discretization step size, objective functions, optimality tolerances and constraints. How to tune all these parameters for the \gls{mpc} scheme is a non-trivial question.

Learning can be an important tool in assisting the tuning process of the \gls{mpc} scheme. In this paper we propose the novel idea of tuning the meta-parameters of the \gls{mpc} scheme using \gls{rl}. By meta-parameters we mean parameters that affect the structure of the \gls{ocp} (including its initial conditions, i.e. when it is computed), rather than parameters that affect the solution to a given \gls{ocp}, the tuning of which has previously been demonstrated in the literature \cite{gros2019data,edwards2021automatic}. This work extends our previous works \cite{bohn2021optimization}, in which we propose to learn when it is necessary to recompute the \gls{mpc} solution, and \cite{bohn2021reinforcement}, in which we suggest to learn the optimal prediction horizon of the \gls{mpc} scheme as a function of the state. Here, we propose a unified framework in which these meta-parameters, as well as any other parameters of the control algorithm, are jointly optimized to simultaneously maximize the control performance and reduce the computational complexity in a configurable manner. The control algorithm we propose consists of a recomputation policy that decides when the \gls{mpc} solution should be computed, a state-feedback controller (i.e. the \gls{lqr}) that is applied on the predicted state trajectory produced by the \gls{mpc} in between \gls{mpc} computations, and \pgls{rl} algorithm that incorporates the parameters of these controllers and optimizes them according to the specified objective. Although we can enforce input constraints, there is no mechanism to ensure that the state constraints hold in our control algorithm, and the constraint satisfaction properties will be determined by the behaviour identified as optimal through \gls{rl}. While we demonstrated in the aforementioned works the effectiveness of learning these meta-parameters in isolation, it is clear that the questions of \emph{when} and \emph{how} (wrt. its tunable parameters) to compute the control algorithm are related, and treating them separately fails to consider the interactions and indirect effects at play.

The contributions of this paper can be summarized as:

\begin{enumerate}
  \item We propose the novel idea of optimizing the meta-parameters of the \gls{mpc} scheme using \gls{rl}.
  \item We develop a novel \gls{mpc} formulation in which the performance and computational power usage are jointly optimized in a configurable, automated manner, which could open up new applications for \gls{mpc}.
  \item To realize the proposed algorithm, we employ novel use of mixture distributions for \gls{rl}.
\end{enumerate}

The rest of the paper is organized as follows. Section \ref{sec:theory} introduces the problem formulation and the necessary theoretical background for our method. Section \ref{sec:method} then describes our method in detail, before we in Section \ref{sec:experiment} demonstrate our method on a numerical example with the inverted pendulum.

\section{Preliminaries and Problem Formulation}\label{sec:theory}
We consider control problems on the form:

\begin{align}
    x_{t+1} &= f(x_t, u_t), \enspace \min_{x, u} \sum_{t=0}^T C_t(x_t, u_t) \label{eq:prob_definition}
\end{align}

where $x$ is the state vector, $u$ is the control input vector, the function $f$ defines the (discrete-time) system dynamics and $C$ defines the cost objective to be minimized. The system runs in an episodic fashion, beginning in some initial state $x_0$ and terminating after some predetermined time $T$ has passed. For any statement that holds regardless of the time, we omit the time subscript. To denote a contiguous sequence of points we use the subscript $x_{t:t+n}$, i.e. the sequence of states from time $t$ to time $t + n$. We denote the time dimension of variables internal to any controller scheme with a subscript $k$. Finally, matrices are denoted with bold uppercase letters, e.g. $\matr{A}$.
\subsection{Control Algorithm}
\subsubsection{Model Predictive Control}\label{sec:theory:mpc}


In this paper we consider adaptive-horizon nonlinear state-feedback discrete-time \gls{mpc} \citep{allgower_nonlinear_1999,rawlings2017model}. The \gls{mpc} receives as arguments the current state of the plant, $\bar{x}_t$, exogenous input variables (e.g. reference signals), $\hat{p}_t$, as well as the prediction horizon, $N_t$, for the \gls{ocp}. We label the \gls{mpc} control law (for a given horizon selected by the horizon policy $\pi_{\theta^\param{N}}^\param{N}$, more on this in Section \ref{sec:method:contpolicies}) as:

\begin{align}
   \left( u^{\param{M}}_{t:t + N_t - 1}, \enspace \hat{x}_{t:t+N_t} \right) = \pi_{\theta^\param{M}}^\param{M}(\bar{x}_t, \hat{p}_t, N_t) 
\end{align}

 where the first return value is the optimal input sequence, the second return value is the predicted optimal state trajectory, and $\theta^\param{M}$ are the tunable parameters of the \gls{mpc} scheme.  The \gls{ocp} is formulated as:

\begin{align}
    \min_{x, u} \quad &\sum_{k=0}^{N_t-1} \rho^k \ell_{\theta^\param{M}}(x_k, u_k, \hat{p}_k) + \rho^{N_t} m_{\theta^\param{M}}(x_{N_t}), \label{eq:mpc:obj}  \\ 
    \textrm{s.t. \quad} &x_0 = \bar{x}_t \label{eq:mpc:init} \\ 
    &x_{k+1} = \hat{f}_{\theta^\param{M}}(x_k, u_k, \hat{p}_k)  \quad \forall \enspace k \in 0, \dots, N_t - 1 \label{eq:mpc:dynamics} \\ 
    &h_{\theta^\param{M}}(x_k, u_k) \leq 0 \qquad \quad \enspace \, \> \forall \enspace k \in 0, \dots, N_t - 1 \label{eq:mpc:h}
\end{align}

Here, $\ell_{\theta^\param{M}}$ is the stage cost consisting of a task-specific objective and the input change term $\Delta u_k^\top \matr{D}_{\theta^\param{M}} \Delta u_k$ where $\matr{D}_{\theta^\param{M}}$ is the input-change weight matrix and $\Delta u_k = u_{k} - u_{k-1}$, $m_{\theta^\param{M}}$ is the terminal cost, $\rho \in (0, 1]$ is the discount factor, $\hat{f}_{\theta^\param{M}}$ is the \gls{mpc} dynamics model, and $h_{\theta^\param{M}}$ is the constraint vector. The stage cost evaluates the computed solution locally up to $N_t-1$ steps, the terminal cost $m_{\theta^\param{M}}(x_{N_t})$ should therefore provide global information about the desirability of the terminal state of the computed trajectory, which helps the \gls{mpc} avoid sub-optimal performance. Therefore, the more accurate the terminal cost is wrt. to the total cost of the infinite horizon solution, the shorter horizons can be used in the \gls{mpc} scheme while still delivering good control performance \citep{lowrey2018plan,zhong_value_2013}. We therefore propose to learn the value function of the \gls{mpc}, which measures the total cost accrued over an infinite horizon, and use it as the terminal cost. For brevity, we refer the reader to \cite{bohn2021reinforcement,lowrey2018plan,zhong_value_2013} for more information on this topic. Note that while we assume state-feedback for simplicity, the formulation could easily be extended with an estimator such as the moving horizon estimator which can be tuned in unison with the \gls{mpc} \cite{esfahani2021reinforcement}. 

We further employ a modification to the \gls{mpc} framework called event-triggered \gls{mpc}, in which the \gls{ocp} is not recomputed at every time step, but rather a recomputation policy $\pi_{\theta^\param{c}}^\param{c}$ decides at every time step whether the \gls{ocp} should be recomputed. Thus, not only the first input of the \gls{mpc} input sequence $u^{\param{M}}_t$ is applied to the plant, but rather a variable number of inputs $u^{\param{M}}_{t:t + n}$, $n < N_t$ are applied sequentially at the corresponding time instance until the recomputation policy triggers the recomputation of the \gls{ocp} at time step $t + n$. We will detail this recomputation policy in Section \ref{sec:method:recomputation}.

\subsubsection{Note on Feasibility}
There is no explicit mechanism in our method that ensures the recursive feasibility of the \gls{mpc} scheme, or the stability of the control algorithm. However, the control algorithm and tuning framework we present is agnostic wrt. to the implementation of the underlying controllers. As such, one could modify the \gls{mpc} scheme by adding e.g. assumptions on the form and magnitude of the disturbances, add terminal constraints, or entirely replace the \gls{mpc} scheme with more complex formulations, e.g. robust \gls{mpc} \cite{bemporad1999robust} or tube \gls{mpc} \cite{mayne2005robust}. 


The prediction horizon is one of the most important tunable parameters to achieve stable control with the \gls{mpc} scheme \citep{grune,mayne2000constrained}. Learning approaches can therefore be an important tool in identifying prediction horizons that yields a stable control system. As discussed in Section \ref{sec:method:rewf}, \gls{rl} is an optimization procedure that seeks optimality  wrt. its reward function, therefore, if it produces non-stabilizing solutions this would suggest that the learning problem itself is ill-posed.

\subsubsection{The Linear Quadratic Regulator}\label{sec:theory:lqrcompat}
The \gls{lqr} \citep{bertsekas1995dynamic} is a state-feedback controller that arises as the optimal solution to unconstrained control problems where the dynamics are linear, and the cost is quadratic. The role of the \gls{lqr} in our control algorithm is to act as the linear feedback correction of the \gls{mpc}, that can be applied to compensate for errors in the open loop predicted state trajectory between \gls{mpc} recomputations. To accomplish this, we employ a time-varying \gls{lqr} where the $\matr{A}_t$ and $\matr{B}_t$ matrices are obtained from the \gls{mpc} scheme as the linearization of the \gls{mpc} model each time it is computed:

\begin{flalign}
    \matr{A}_{t+1:t+N_t}, \matr{B}_{t+1:t+N_t} &= \mathrm{linearize}(\hat{f}_{\theta^\param{M}}) \big\rvert_{x_{t+1:t+N_t}, u_{t:t+N_t-1}} & \\
    \matr{A}, \matr{B} &= \mathrm{linearize}(\hat{f}_{\theta^\param{M}}) \big\rvert_{x^s_t, u^s_t} \label{eq:lqr_dyn_mat:horizonend}
\end{flalign}

After the horizon end (i.e. when $t > i + N_i$ where $i$ is the time of last \gls{mpc} computation), we set the \gls{lqr} dynamics matrices as the time-invariant matrices corresponding to the steady state (equilibrium) of the system $x^s_t, u^s_t$ \eqref{eq:lqr_dyn_mat:horizonend}. The specific linearization procedure depends on the implementation of the dynamics model in the \gls{mpc}, i.e. discrete vs continuous model and the accompanying discretization schemes. The $\matr{Q}$ and $\matr{R}$ cost-weighting matrices are initialized from the \gls{mpc} objective as follows, and then tuned further as described in Section \ref{sec:method}.
 
\begin{align}
    \matr{Q}_\mathrm{init} \leftarrow \frac{\partial^2 \ell_{\theta^\param{M}}(x, u, \hat{p})}{\partial x^2}|_{x^s_0, u^s_0}, \enspace    \matr{R}_\mathrm{init} \leftarrow \frac{\partial^2 \ell_{\theta^\param{M}}(x, u, \hat{p})}{\partial u^2}|_{x^s_0, u^s_0}
\end{align}

In this paper we focus on the discrete-time formulation of the \gls{lqr}. The \gls{lqr} control law \eqref{eq:lqr:control_law_tv} consists of a state-feedback gain matrix $\matr{K}$ that is derived from the dynamics matrices $\matr{A}$ and $\matr{B}$, the cost-weighting matrices $\matr{Q}$, $\matr{R}$ and $\matr{N}$, and the solution $\matr{S}$ to the discrete-time Riccati-equation \eqref{eq:lqr:dare}:

\begin{align}
    \begin{split}
       \matr{S}_{k} &= \matr{Q}_k + \matr{A}_k^\top \matr{S}_{k+1} \matr{A}_k - (\matr{A}_k^\top \matr{S}_{k+1} \matr{B}_k + \matr{R}_k) \\
       & \quad \> (\matr{R}_k + \matr{B}_k^\top \matr{S}_{k+1} \matr{B}_k)^{-1} (\matr{B}_k^\top \matr{S}_{k+1} \matr{A}_k + \matr{R}_k^\top) \label{eq:lqr:dare}
    \end{split} \\ 
    \matr{K}_k &= -(\matr{R}_k + \matr{B}_k^\top \matr{S}_{k+1} \matr{B}_k)^{-1} (\matr{B}_k^\top \matr{S}_{k+1} \matr{A}_k + \matr{R}_k^\top) \label{eq:lqr:K} \\
    u^{\param{L}}_k(x_k) &= \pi^{\param{L}}_{\theta^{\param{L}}}(x_k) = -\matr{K}_k x_k \label{eq:lqr:control_law_tv}
\end{align}

We consider in this paper the vector $\theta^\param{L}$ containing all the elements of the matrices $\matr{Q}, \matr{R}, \matr{N}$ as parameters of the \gls{lqr} controller that can be optimized. The finite horizon \gls{lqr} control problem is stated as:
\begin{align}
    \min_{x, u} \quad &\sum_{k=0}^{N_i - 1} \left[x_{k}^\top \matr{Q} x_{k} + u_k^\top \matr{R} u_k + 2 x_k^\top \matr{N} u_k \right],  \label{eq:lqr:finite_obj} \\
    \textrm{s.t. } \quad &\matr{Q}_k - \matr{N}_k \matr{R}_k^{-1} \matr{N}_k^\top \succ 0, \> \matr{R}_k > 0 \\
    &x_{k+1} = \matr{A}_k x_k + \matr{B}_k u_k \label{eq:lqr:tv_dynamics}
\end{align}

where $N_i$ is the optimization horizon. We also consider the infinite horizon problem, i.e. $N_i=\infty$, in which the matrices are constant wrt. time and the \gls{dare} is solved at its stationary point, yielding the solution $\matr{S}_\infty$ and consequent $\matr{K}_\infty$.

\subsection{Reinforcement Learning}
The system to be optimized in the \gls{rl} framework \cite{sutton_reinforcement_2018} is formulated as a \gls{mdp}. The \gls{mdp} is defined by a state space $s \in \mathcal{S}$, a set of actions $a \in \mathcal{A}$, a transition probability matrix $\mathcal{T}$ that governs the evolution of states as a function of time and actions, i.e. $s_{t+1} = \mathcal{T}(s_t, a_t)$, a reward function $R(s, a)$ that describes the desirability of the states of the system, and finally, the discount factor $\gamma \in [0, 1)$ (note the different limits from $\rho$) that describes the relative importance of immediate and future rewards. Note that rewards, $R$, are interchangeable with costs, $C$, through the substitution $R = -C$ and changing maximization of the objective to minimization. We will use rewards when discussing \gls{rl} as this is customary in \gls{rl}.

The objective in \gls{rl} is to develop a policy $\pi_\theta$, i.e. a function that maps from states to actions (here parameterized by $\theta$) that maximizes the expected discounted sum of rewards. In this paper we use the policy gradient algorithm \gls{ppo} \cite{schulman2017proximal}. For brevity, we refer the reader to the original paper \cite{schulman2017proximal} for details on how \gls{ppo} updates the parameters to be optimized. In general, policy gradient algorithms updates a parameterized stochastic policy directly in the parameter space, by sampling actions from the policy's action distribution and observing the outcomes in terms of states and rewards. Parameters are adjusted to increase the likelihood of actions leading to high rewards using gradient ascent with gradients from the policy gradient theorem \cite{sutton_reinforcement_2018}:

\begin{equation}
\theta \leftarrow \theta_{\mathrm{old}} + \eta \nabla_\theta J(\theta)
\label{eq:rl:gd}
\end{equation}
where
\begin{align}
    &\nabla_\theta J(\theta) = \mathbb{E}\left[\sum_{t=0}^T \nabla_\theta \log \pi_\theta (a_t | s_t) G(t)\right] \label{eq:rl:pgt} \\
    & G(t) = \sum_{t'=t}^{T} \gamma^{t'-t}R(s_{t'}, \pi_\theta(s_{t'})) 
\end{align}

\section{Method}\label{sec:method}
\subsection{A State Space with the Markov Property}\label{sec:method:markov_state}
For \gls{rl} theory to hold for the control system we wish to optimize, the state space needs to have the Markov property, i.e. future states should not depend upon past states given the current state. This section outlines such a Markovian state.

Consider the input sequence $u^{\param{M}}_{i:i + N_i - 1}$ and the predicted state trajectory $\hat{x}_{i + 1:i+N_i}$ computed by the \gls{mpc} at time $i$. As discussed in Section \ref{sec:theory:mpc}, a variable number $n \in \{0, 1, \dots, N_i-1\}$ of these inputs are applied to the plant. Since $n$ is not known a priori, the state vector $x$ is not a sufficient state representation to yield the Markov property for the control system consisting of the recomputation policy, the horizon policy, and the \gls{mpc} and \gls{lqr} control laws. We define an augmented state $s$ in \eqref{eq:state_s} that contains the current plant state and exogenous variables, labeled $\bar{x}_t$ and $\hat{p}_t$, the state of the system, exogenous variables, and prediction horizon used when the last \gls{mpc} computation took place, labeled $\bar{x}_i$, $\hat{p}_i$, and $N_i$, as well as the number of time steps since the \gls{mpc} computation, $t - i$:

\begin{align}
    s_t &= \left[\bar{x}_i, \hat{p}_i, N_i, \bar{x}_t, \hat{p}_t, t - i \right]^\top \label{eq:state_s}
\end{align}

where $\bar{x}_t$ has evolved from $\bar{x}_i$ according to the real system dynamics. When the \gls{mpc} problem is recomputed, the deterministic transition $i \leftarrow t, \bar{x}_i \leftarrow \bar{x}_t, \hat{p}_i \leftarrow \hat{p}_t, N_i \leftarrow N_t$ takes place. The \gls{mpc} and \gls{mpc} plus \gls{lqr} control laws deployed on the plant is then \eqref{eq:cl:mpc} and \eqref{eq:cl:mpclqr}, respectively, while the total control system is defined as \eqref{eq:cl:total}, where the $\mathrm{proj}_{h_{\theta^\param{M}}}$ operator projects the control input onto the constraint vector $h_{\theta^\param{M}}$ \eqref{eq:mpc:h}:

\begin{align}
    \pi_{\theta^\param{M}}^{\param{M}}(s_t) &= \begin{cases}u^{\param{M}}_{t - i}(\bar{x}_i, \hat{p}_i, N_i)&, \mathrm{ if } \enspace t - i < N_i \\ 0 &, \mathrm{ otherwise}\end{cases} \label{eq:cl:mpc} \\
    \pi_{\theta^\param{M,L}}^{\param{ML}}(s_t) &= \begin{cases}u^{\param{M}}_{t - i}(\bar{x}_i, \hat{p}_i, N_i) + u^{\param{L}}_{t}(\hat{x}_{t - i} - \bar{x}_t), \mathrm{ if } \enspace t - i < N_i \\
                               u^{\param{L}}_t(x^s_t - \bar{x}_t)\qquad \qquad \qquad \qquad \enspace,  \mathrm{ if } \enspace t - i \geq N_i
                               \end{cases}
    \label{eq:cl:mpclqr} \\
    \pi_{\theta^\param{M,L}}^{\param{CS}}(s_t) &= \mathrm{proj}_{h_{\theta^\param{M}}} \left(\pi_{\theta^\param{M}}^{\param{M}}(s_t) + \pi_{\theta^\param{M,L}}^{\param{ML}}(s_t)\right) \label{eq:cl:total}
\end{align}

and where $x^s_t$ is the steady state equilibrium of the system at time $t$, which is time-variant in the case of time-varying references.

\begin{assumption}
The time-varying exogenous input variables $\hat{p}_t$ are generated by a Markovian process.
\end{assumption}

\begin{proposition}
The state $s$ has the Markov property, i.e. $P(s_{t + 1} | s_t) = P(s_{t+1} | s_{0:t})$
\end{proposition}

\begin{proof}
First, note that by definition from \eqref{eq:prob_definition} $x$ is Markovian given $u$:
\begin{align}
    x_{t+1} &= f(x_t, u_t) \\
    \Rightarrow P(x_{t+1} | x_t, u_t) &= P(x_{t+1} | x_{0:t}, u_{0:t})
\end{align}

However, the control law $\pi_{\theta^\param{M,L}}^{\mathrm{CS}}$ \eqref{eq:cl:total} that determines $u_t$ consists of $\pi_{\theta^\param{M}}^{\param{M}}$ \eqref{eq:cl:mpc} --- which depends on the state and exogenous variables $\bar{x}_i, \hat{p}_i$ and the prediction horizon $N_i$ at the last \gls{mpc} computation --- and $\pi_{\theta^{\param{ML}}}^{\param{M,L}}$ \eqref{eq:cl:mpclqr}, which depends on the current state $x_t$ and the $(t-i)^{\mathrm{th}}$ element of the \gls{mpc}'s predicted state trajectory, and $x^s_t$ which depends on $\hat{p}_t$. Therefore, with $s$ as defined in \eqref{eq:state_s} we have:

\begin{align}
    P(u_t | \pi_{\theta^\param{M,L}}^{\mathrm{CS}}, s_t) = P(u_t | \pi_{\theta^\param{M,L}}^{\mathrm{CS}}, s_{0:t})
\end{align}

Finally, note that the \gls{mpc} input sequence $u^{\param{M}}_{i+1:i + N_i - 1}$ and the predicted state trajectory $\hat{x}_{i:i+N_i}$ follows from $\pi^{\param{M}}_{\theta^{\param{M}}}$ and its arguments (which are all known), and as such does not need to be contained in $s$ for $s$ to be Markovian:

\begin{align}
    \Rightarrow P(s_{t + 1} | s_t) = P(s_{t+1} | s_{0:t})
\end{align}


\end{proof}

\begin{figure}[h]
    \centering
    \includegraphics{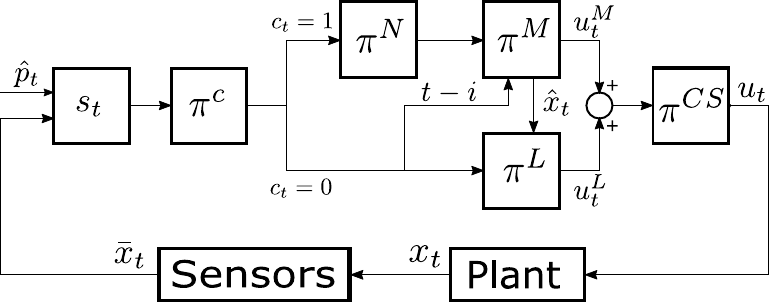}
    \caption{An overview of the control algorithm. Not shown here is the connection of each policy's output to the \gls{rl} algorithm that updates the policy's parameters.}
    \label{fig:control_algorithm}
\end{figure}

\subsection{Policies}\label{sec:method:policies}
We use the notation $\pi_{\theta}(s)$ to label a deterministic function of the state $s$ that is parameterized by $\theta$, and $\pi_{\theta}(\cdot | s)$ to label the stochastic version of the same function. The notation $\tilde{\cdot}$ signifies that the value is drawn from the policy's probability distribution. The implementation of the \gls{ppo} algorithm requires the log probability of the policies, and as such we will list the probability distributions, denoted $P$, and corresponding log probabilities, denoted $\log P$, for the policies we optimize. For brevity, we omit listing input-independent conditional variables as arguments (e.g. covariance) of the distributions.

\subsubsection{The Recomputation Policy}\label{sec:method:recomputation}
The recomputation policy $\pi_{\theta^\param{c}}^\param{c}$ decides on each step whether the \gls{mpc} problem should be recomputed, or if the previously computed solution is still acceptable. In other words, it is a binary variable that chooses among the two different options: recompute or not. As such, we model it as a Bernoulli-distributed random variable, where the policy outputs the logit of the probability that the \gls{ocp} should be recomputed \eqref{eq:bernoulli:logit}, from which we can deduce the probability of not recomputing. We label the output of the of the recomputation policy $c \in \{0, 1\}$. The Bernoulli distribution has the \gls{pmf} \eqref{eq:bernoulli:pmf}.


\begin{align}
    w &= \frac{1}{1 + \exp\left(-\pi_{\theta^\param{c}}^\param{c}(s)\right)} \label{eq:bernoulli:logit} \\
    P^\param{c}(c | s) &= \begin{cases}
            1 - w&, \textrm{ if } c = 0 \\
            w&, \textrm{ if } c = 1
        \end{cases} \label{eq:bernoulli:pmf}
\end{align}

The log probability can be expressed as follows:

\begin{align}
    \log P^\param{c}(c| s) = c \log (w) + (1 - c) \log(1 - w)
\end{align}

Finally, the policy gradient of the recomputation policy is:

\begin{align}
    &\nabla_{\theta^\param{c}} \log \pi^\param{c}_{\theta^\param{c}}(c | s) = \nabla_{\theta^\param{c}}\left(c \log(w) + (1-c) \log(1 - w) \right) \\
    &= -c\nabla_{\theta^\param{c}} \pi_{\theta^\param{c}}^\param{c}(s) w + \nabla_{\theta^\param{c}} \pi_{\theta^\param{c}}^\param{c}(s)(1-c)\left(-1-w\right) \\
    &= \nabla_{\theta^\param{c}} \pi_{\theta^\param{c}}^\param{c}(s)\left(c - 1 - w\right) \label{eq:logreg_grad}
\end{align}

\subsubsection{The Horizon Policy}
The \gls{mpc} prediction horizon $N$ is a positive integer, which we model with the GP-2 variant of the \gls{gpd} \cite{wang1997modeling}. Other models we considered but decided against includes categorical classification models, as they do not consider the ordinal information of the horizon variable, and the standard Poisson distribution, which we found to be too inflexible due to its mean and variance being equal. The horizon model outputs the rate parameter of the \gls{gpd}, $\mu$, as a function of $s$, while the dispersion parameter $\alpha$ is learned as an input-\emph{independent} variable.

\begin{align}
    P^\param{N}(N | s) &= \left(\frac{\mu}{1 + \alpha \mu}\right)^\param{N} \frac{(1 + \alpha N)^{N-1}}{N!}e^{\left(-\frac{\mu (1 + \alpha N)}{1 + \alpha \mu}\right)} \label{eq:gpd:pmf} \\
    \mathbb{E}(N) &= \mu, \enspace \mathbb{V}(N) = \mu (1 + \alpha \mu)^2 \label{eq:gpd:moments} 
\end{align}

In this model, $\alpha$ is restricted according to $1 + \alpha \mu > 0$ and $1 + \alpha N > 0$. To address these constraints, we introduce two hyperparameters $N_{\mathrm{min}}$ and $N_{\mathrm{max}}$ that correspond to the minimum and maximum horizons that the \gls{mpc} should operate with. To constrain the rate parameter $\mu$ we apply the hyperbolic tangent function, denoted $\tanh \in [-1, 1]$, and then linearly scale it to the limits defined by $N_{\mathrm{min}}$ and $N_{\mathrm{max}}$. Finally, the learned $\alpha$ parameter is clipped according to restrictions outlined above:

\begin{align}
    \pi^\param{N}_{\theta^\param{N}}(s) &= \mathrm{scale}\left(\tanh(\mu), N_{\mathrm{min}}, N_{\mathrm{max}}\right) \\
    \alpha_{\mathrm{new}} &= \max\left(\alpha, -\frac{1}{N_{\mathrm{max}}}\right)
\end{align}

The log probability of the \gls{gpd} and the policy gradient of the horizon policy is defined as follows:
\begin{align}
\begin{split}
    &\log P^\param{N}(N | s) = N \log\left(\frac{\mu}{1 + \alpha \mu}\right) \\ 
    &+ (N - 1) \log(1 + \alpha N) - \frac{\mu ( 1 + \alpha N)}{1 + \alpha \mu} - \log(N!)
\end{split} \\
\begin{split}
    &\nabla_{\theta^\param{N}} \log \pi^\param{N}_{\theta^\param{N}}(N | s) = \nabla_{\theta^\param{N}} \pi^\param{N}_{\theta^\param{N}}(s) \Bigg(\frac{N}{\pi^\param{N}_{\theta^\param{N}}(s)} + \\ 
    &\frac{\alpha}{1 + \alpha \pi^\param{N}_{\theta^\param{N}}(s)} +
        \frac{1 + \alpha N + \pi^\param{N}_{\theta^\param{N}}(s)(\alpha(1+N + \alpha N) + 1)}{1+2\alpha \pi^\param{N}_{\theta^\param{N}}(s) + (\alpha \pi^\param{N}_{\theta^\param{N}}(s))^2} \Bigg)
\end{split}
\end{align}

There are several techniques to sample from this distribution \cite{demirtasAccurate2017}. We favor the normal approximation sampling technique for its low run-time complexity: With a sufficiently high rate parameter (i.e. $\mu \gtrapprox 10$), the \gls{gpd} is approximately normally distributed with mean and variance as given in \eqref{eq:gpd:moments}:

\begin{align}
\begin{split}
    \pi_{\theta^\param{N}}^\param{N}(N | s) \approx \mathrm{clip}(&\floor{\mu + \sqrt{\mu ( 1 + \alpha \mu)^2} \zeta + 0.5}, \\
    &N_{\mathrm{min}}, N_{\mathrm{max}}), \enspace \zeta \sim \mathcal{N}(0, 1)
\end{split}
\end{align}

When optimizing the horizon policy in isolation (that is, not together with the rest of the control system) we find that faster convergence and more stable solutions are obtained by learning on an augmented \gls{mdp} where every action selected by the policy is repeated $d > 1$ times. Hence the policy only senses every $d$'th state, and the rewards it receives are the cumulative reward over every step of the control system in the $d$ steps. This technique is known as "frame-skip" in \gls{rl} and is an effective method to enhance learning for problems with discrete actions, see e.g. learning to play atari-games \cite{mnih2015human}, but also for continuous control \cite{kalyanakrishnan2021analysis}. While the exact mechanisms behind the improvements stemming from frame-skipping is not fully understood, it is clear that in certain problems it increases the signal-to-noise ratio of every data sample, which simplifies the credit assignment problem. When using $d > 1$ during exploration, we use $d = 1$ during exploitation (evaluation) as this increases performance. We find that the horizon policy is not sensitive to the exact value of $d$, and all values in $d \in [2, T]$ generally accelerates learning. Finally, note that when optimizing the horizon policy and the recomputation policy together, the latter electing to not recompute the \gls{mpc} will naturally enforce a form of frame-skipping.

\subsubsection{The Controller Policies}\label{sec:method:contpolicies}
To ensure sufficient exploration we formulate a stochastic version of the \gls{mpc}, and the \gls{mpc} plus \gls{lqr} control laws (denoted $\pi_{\theta^{\param{M,L}}}^{\param{ML}}(s) = \pi_{\theta^{\param{M}}}^{\param{M}}(s) + \pi_{\theta^{\param{L}}}^{\param{L}}(s)$), modeling them as Gaussian random variables \eqref{eq:gauss:cl}. The mean is then the output of the control laws, and the covariance of each controller is a learned input-\emph{independent} variable of the \gls{rl} algorithm. To be concise, we use $*$ in place of the superscript for the control laws as in \eqref{eq:gauss:cl}. We will first develop these policies for the \gls{mpc} scheme with a given horizon $N$, where the output is made stochastic as follows:

\begin{align}
    \pi_{\theta^*}^*(u^* | s, N) &= \mathcal{N}\left(\pi_{\theta^*}^*(s, N), \Sigma^*\right), \enspace * \in \{\param{M},\param{ML}\} \label{eq:gauss:cl} \\
    P^\param{u}(u^* | s, N) &= \frac{1}{\Sigma^* \sqrt{2 \pi}} \mathrm{exp}\left(-\frac{1}{2} \left(\frac{u^* - \pi_{\theta^*}^*(s, N) }{\Sigma^*}\right)^2\right)
\end{align}

Note that the argument $N$ to the policies is redundant (and thus these policies are equivalent to those in Section \ref{sec:method:markov_state}), as $N_i$ in $s$ will always reflect the latest $N$ that is decided by the horizon policy ahead of control law calculation. We use this argument here to clarify the derivations. The log probability of the distribution and the policy gradient is:

\begin{align}
\begin{split}
    \log P^\param{u}(u^* | s, N) =& -\frac{1}{2 \Sigma^{*^2}}\left(u^* - \pi_{\theta^*}^*(s, N)\right)^2 - \\
    &\frac{1}{2}\log(\Sigma^{*^2}) - \frac{1}{2} \log(2 \pi)
\end{split} \\
    \nabla_{\theta^*} \log \pi_{\theta^{*}}^*(u^* | s, N) =& \Sigma^{*^{-1}} (\pi^*_{\theta^*}(s, N) - u^*) \nabla_{\theta^*} \pi_{\theta^*}^*(s, N)
\end{align}

Then, note that the adaptive-horizon \gls{mpc} control law is a distribution of the \gls{mpc} schemes over the range of prediction horizons, with weights $P^\param{N}$ assigned by the horizon policy. We first query the horizon policy for which prediction horizon to employ, and then the solution to the \gls{mpc} problem \eqref{eq:mpc:obj} is computed with the selected horizon:

\begin{align}
    \pi_{\theta^{\param{M,N}}}^{\param{MN}}(s) = \pi_{\theta^{\param{M}}}^{\param{M}}\left(x_t, \hat{p}_t, \pi_{\theta^\param{N}}^\param{N}(s)\right)
\end{align}

where the superscript $N$ indicates that this is the adaptive-horizon \gls{mpc} policy. Using the indicator function:

\begin{align}
    \mathbbm{1}_{A} = \begin{cases}1, &\textrm{ if } A \\ 0, &\textrm{ otherwise}\end{cases}
\end{align}

we can formulate the probability distributions and log distributions for the stochastic adaptive-horizon control policies:

\begin{align}
    &P^\param{Nu}(u^{*} | s) = \sum_{N=N_{\mathrm{min}}}^{N_{\mathrm{max}}} P^\param{N}(N | s) P^\param{u}(u^*|s, N), * \in \{\param{MN},\param{MLN}\} \\
    &\log P^\param{Nu}(u^{*} | s) = \log\left(\sum_{N=N_{\mathrm{min}}}^{N_{\mathrm{max}}}P^\param{N}(N|s) P^\param{u}(u^*|s, N)\right) \\
    &= \sum_{N=N_{\mathrm{min}}}^{N_{\mathrm{max}}}\mathbbm{1}_{N=\tilde{N}}\left( \log P^\param{N}(N|s) + \log P^\param{u}(u^*|s, N)\right) \label{eq:cp:logPsum} \\
    &= \log P^\param{N}(\tilde{N}|s) + \log P^\param{u}(u^*|s, \tilde{N})\label{eq:cp:logP}
\end{align}

where the $\log$ operator can be applied inside the summation over the prediction horizons in \eqref{eq:cp:logPsum}, because we know the value $\tilde{N}$ of the horizon variable $N$ that is sampled in advance of the calculation of the control laws \cite{friedman2001elements}.

\subsubsection{The Complete Policy}
We collect all the parameters of the meta-parameter-deciding recomputation and horizon policies described above, and the parameters of the controllers into a single parameter vector $\theta = \left[\theta^\param{c}, \theta^\param{N}, \alpha, \theta^{\param{M}},\theta^{\param{L}}, \Sigma^\param{M}, \Sigma^{\param{ML}}\right]^\top$, and define the complete policy $\pi_\theta$, whose input is the state $s$ and output is the action $a = \left[c, N, u^{\param{M}}, u^{\param{ML}}\right]$

Section \ref{sec:method:recomputation} presents the recomputation problem of the \gls{mpc} and the policy that decides when to compute it. We now view the recomputation policy $\pi_{\theta^\param{c}}^\param{c}$ as selecting the active controller between the two control laws. It is then clear that $\pi_\theta$ is a mixture distribution between the Gaussian control policies where the weights of the mixture are assigned by the Bernoulli recomputation policy $\pi_{\theta^\param{c}}^\param{c}$. We label the probability distribution of the complete policy $\pi_\theta$ as $P^\param{a}$ and define it and the log probability as follows:

\begin{align}
    &P^\param{a}(\tilde{a} | s) = P^\param{c}(0 | s)P^\param{Nu}(\tilde{u}^{\param{MLN}} | s) +  P^\param{c}(1 | s)P^\param{Nu}(\tilde{u}^{\param{MN}} | s) \\
    \begin{split}
        &\log P^\param{a}(\tilde{a} | s) = \mathbbm{1}_{\tilde{c} = 0}\left(\log P^\param{c}(0 | s) + \log P^\param{Nu}(\tilde{u}^{\param{MLN}} | s)\right) + \\ & \qquad \qquad \qquad \, \mathbbm{1}_{\tilde{c} = 1}\left( \log P^\param{c}(1 | s) + \log P^\param{Nu}(\tilde{u}^{\param{MN}} | s)\right)
    \end{split} \label{eq:policy:step1}  \\
    \begin{split}
        &=\mathbbm{1}_{\tilde{c} = 0}\left(\log P^\param{c}(0 | s) + \log P^\param{N}(\tilde{N}_i | s) + \log P^\param{u}(\tilde{u}^{\param{ML}} | s, \tilde{N}_i)\right) \\
        & \quad \> \, + \mathbbm{1}_{\tilde{c} = 1} \left( \log P^\param{c}(1 | s) + \log P^\param{N}(\tilde{N} | s) + \log P^\param{u}(\tilde{u}^{\param{M}} | s, \tilde{N}) \right)
    \end{split} \label{eq:policy:logP}
\end{align}

where we again used the fact that we know the values of the sampled variables $\tilde{c}, \tilde{N}$ to take the logarithm of the sums in \eqref{eq:policy:step1} and  \eqref{eq:policy:logP}, and $\tilde{N}$ represents the output of the horizon policy at the current step $t$ in the case the recomputation policy signals to recompute the \gls{mpc}. While \eqref{eq:policy:logP} is all that is strictly needed to use the \gls{ppo} algorithm, we derive the policy gradient as well to highlight the connection to the \gls{lqr} gradient presented in Section \ref{sec:method:lqropt}, as well as for future applications of this control framework with other \gls{rl} algorithms. 

\begin{align}
    &\nabla_\theta \log \pi_\theta(\tilde{a} | s) = \nabla_\theta \log P^\param{a}(\tilde{a} | s) \\
    \begin{split}
        &= \mathbbm{1}_{\tilde{c} = 0} \Big( \nabla_{\theta^\param{c}} \log \pi_{\theta^\param{c}}^\param{c}(0 | s) + \nabla_{\theta^\param{N}} \log \pi_{\theta^\param{N}}^\param{N}(\tilde{N}_i | s) \\ 
        &\quad \enspace + \nabla_{\theta^{\param{M,L}}} \log \pi_{\theta^{\mathrm{M, L}}}^{\param{ML}}(\tilde{u}^{\param{ML}} | s, \tilde{N}) \Big) + \\
        &\quad \> \, \mathbbm{1}_{\tilde{c} = 1} \Big ( \nabla_{\theta^\param{c}} \log \pi_{\theta^\param{c}}^\param{c}(1 | s) + \nabla_{\theta^\param{N}} \log \pi_{\theta^\param{N}}^\param{N}(\tilde{N} | s) \\ 
        &\quad \enspace + \nabla_{\theta^{\param{M}}}\log \pi_{\theta^{\param{M}}}^{\param{M}}(\tilde{u}^{\param{M}} | s, \tilde{N})  \Big )
    \end{split} \\
    \begin{split}
        &= \mathbbm{1}_{\tilde{c} = 0} \Big( \nabla_{\theta^\param{c}} \log \pi_{\theta^\param{c}}^\param{c}(0 | s) + \nabla_{\theta^\param{N}} \log \pi_{\theta^\param{N}}^\param{N}(\tilde{N}_i | s) +  (\Sigma^{\param{ML}})^{-1} \\
        & \quad \enspace(\pi^{\param{ML}}_{\theta^\param{M,L}}(s, \tilde{N}) - \tilde{u}^\param{ML})(\nabla_{\theta^\param{M}} \pi_{\theta^\param{M}}^\param{M}(s, \tilde{N}_i) + \nabla_{\theta^\param{L}} \pi_{\theta^\param{L}}^\param{L}(s)) \Big) \\
        & \> +  \mathbbm{1}_{\tilde{c} = 1} \Big(
        \nabla_{\theta^\param{c}} \log \pi_{\theta^\param{c}}^\param{c}(1 | s) + \nabla_{\theta^\param{N}} \log \pi_{\theta^\param{N}}^\param{N}(\tilde{N} | s) \\
        &\quad \enspace + (\Sigma^{\param{M}})^{-1}(\pi^{\param{M}}_{\theta^\param{M}}(s, \tilde{N}) - \tilde{u}^\param{M}) \nabla_{\theta^\param{M}} \pi_{\theta^\param{M}}^\param{M}(s, \tilde{N}) \Big)
    \end{split} \label{eq:method:pg}
\end{align}

With this policy one can optimize any parameter of the \gls{mpc} and \gls{lqr} controllers jointly, by providing the gradient of these controllers wrt. the parameters. In the \gls{rl} context, we show how one can obtain these gradients for the \gls{lqr} controller in Section \ref{sec:method:lqropt}. For the gradient of the \gls{mpc} see e.g. \cite{gros2019data,gros_reinforcement_2021}.


\subsection{Optimizing LQR with Reinforcement Learning}\label{sec:method:lqropt}

\subsubsection{The Time-Invariant Case}

To apply the policy-gradient theorem to tune the \gls{lqr}, we need to compute the gradients of the $\matr{K}$-matrix wrt. to the $\matr{Q}$, $\matr{R}$ and $\matr{N}$ matrices. As equations \eqref{eq:lqr:dare} and \eqref{eq:lqr:K} show, the feedback matrix $\matr{K}$ is only implicitly defined in terms of the these matrices, and so direct differentiation of equations \eqref{eq:lqr:dare} and \eqref{eq:lqr:K} is not possible. Since we obtain the system matrices $\matr{A}$ and $\matr{B}$ directly from the \gls{mpc} scheme and they have a clear purpose in making the \gls{mpc} and \gls{lqr} objectives compatible, we assume that they are fixed wrt. $\matr{Q}, \matr{R}, \matr{N}$ such that $\nabla_{\matr{Q}, \matr{R}, \matr{N}} \matr{A} = \nabla_{\matr{Q}, \matr{R}, \matr{N}} \matr{B} = 0$:

\begin{assumption}\label{assumption:dare}
The weighting matrices $ \{\matr{Q}, \matr{R}, \matr{N}\}$ values are such that the \gls{dare} has a solution $\matr{S}_\infty$.
\end{assumption}
\begin{proposition}
We flatten the matrices $\matr{S}, \matr{K}, \matr{Q}, \matr{R}$ and $\matr{N}$ into vectors, and organize them as follows: $y = \{\matr{S}, \matr{K}\}$ and $z = \{\matr{Q}, \matr{R}, \matr{N}\}$, such that the \gls{dare} and $\matr{K}$-matrix equations can be written on the vector form $F(y, z) = 0$. The gradient of $\matr{K}$ wrt. the weighting matrices $\matr{Q}, \matr{R}, \matr{N}$ can then be found as: 
\begin{align}
\frac{\partial y}{\partial z} = -\frac{\partial F}{\partial y}^{-1}  \frac{\partial F}{\partial z } = \nabla_{\matr{Q}, \matr{R}, \matr{N}} \matr{S}_\infty, \matr{K}_\infty \label{eq:ift}
\end{align}
\end{proposition}

\begin{proof}
We rewrite \eqref{eq:lqr:dare}-\eqref{eq:lqr:K} on the general vector form $F(y, z) = 0$ where $y = \{\matr{S},\matr{K}\}$ and $z = \{\matr{Q},\matr{R},\matr{N}\}$.
\begin{align}
    \matr{A}^\top \matr{S} \matr{A} - \matr{S} - (\matr{A}^\top \matr{S} \matr{B} + \matr{N})\matr{K} + \matr{Q} &= 0 \\
    (\matr{B}^\top \matr{S} \matr{B} + \matr{R})\matr{K} - (\matr{B}^\top \matr{S} \matr{A} + \matr{N}^\top) &= 0
\end{align}
We then apply the \gls{ift} which states that:
\begin{align}
    \frac{\partial F}{\partial y}\frac{\partial y}{\partial z} +  \frac{\partial F}{\partial z } = 0 \enspace \Rightarrow \enspace \frac{\partial y}{\partial z} = -\frac{\partial F}{\partial y}^{-1}  \frac{\partial F}{\partial z}  \label{eq:lqr:infty_g}
\end{align}
These gradients are easily obtained with automatic differentiation software.
\end{proof}

Assumption \ref{assumption:dare} holds if $\matr{Q} - \matr{N} \matr{R}^{-1} \matr{N}^\top \succ 0$ and $\matr{R} > 0$, and further it is required that the symplectic pencil of the problem has eigenvalues sufficiently far from the unit circle, which is satisfied if the pair ($\matr{A}, \matr{B}$) is stabilizable and the pair ($\matr{A}$, $\matr{Q}$) is detectable \cite{laub1979schur}. These conditions can be ensured by estimating the gradients as described above, and then solving \pgls{sdp} using the estimated gradients subject to these constraints. However, for simplicity and to avoid constrained optimization, we set $\matr{N} = 0$, simplifying the constraints on the positive (semi)-definiteness to $\matr{Q} \succ 0$ and $\matr{R} > 0$. Further, we write $\matr{Q}$ and $\matr{R}$ in terms of their Cholesky decompositions: $\matr{Q} = \matr{Q}_C^\top \matr{Q}_C$, $\matr{R} = \matr{R}_C^\top \matr{R}_C$, and let the \gls{rl} algorithm adjust the elements of $\matr{Q}_C$ and $\matr{R}_C$. 

\subsubsection{Time-Varying Gradients}
In what follows we will assume $\matr{N} = 0$, for simplicity of the derivation and expressions (and the constraint reasons outlined above). We find the gradients with respect to $p$, where $p$ is an arbitrary scalar element of the $\matr{Q}$ and $\matr{R}$ matrices. The full gradients $\nabla_{\matr{Q}, \matr{R}} \matr{K}_k$ can then be found by solving \eqref{eq:lqr:K_tv_g} for each parameter of $\matr{Q}$ and $\matr{R}$ and arranging the results into the appropriate matrix structure. The derivations involve repeated application of the chain rule and the following relation for the gradient of the matrix inverse, where we define $\matr{E}$ for convenience:

\begin{align}
    \matr{E}_k &= \matr{R} + \matr{B}^\top_k \matr{S}_{k+1} \matr{B}_k, \enspace \nabla_p \matr{E}^{-1} = -\matr{E}^{-1} \nabla_p \matr{E} \matr{E}^{-1}
\end{align}

\begin{align}
    \begin{split}
        \nabla_{p} \matr{S}_{k} =& \nabla_{p} \matr{Q} + \matr{A}_k^\top (-\nabla_{p} \matr{S}_{k+1} \matr{B}_k \matr{E}_k^{-1} \matr{B}_k^\top \matr{S}_{k+1} \matr{A}_k \\
        &+\matr{S}_{k+1} \matr{B}_k(\matr{E}_k^{-1} \nabla_{p} \matr{E}_k \matr{E}_k^{-1} \matr{B}_k^\top \matr{S}_{k+1} \matr{A}_k\\ 
        &- \matr{E}_k^{-1} \matr{B}_k^\top \nabla_{p} \matr{S}_{k+1} \matr{A}_k) + \nabla_{p} \matr{S}_{k+1} \matr{A}_k)
    \end{split} \label{eq:lqr:S_tv_g}
     \\
     \begin{split}
         \nabla_{p} \matr{K}_k =& -\matr{E}_k^{-1} \nabla_p \matr{E}_k \matr{E}_k^{-1} \matr{B}_k^\top \matr{S}_{k+1} \matr{A}_k + \matr{E}_k^{-1} \matr{B}_k^\top \nabla_{p} \matr{S}_{k+1} \matr{A}_k \label{eq:lqr:K_tv_g}
     \end{split}
\end{align}

Note that these gradients could easily be extended with time-varying $\matr{Q}$ and $\matr{R}$ matrices. The time-varying gradients of $\matr{K}_k$ are given as follows:

\begin{align}
    \nabla_{\matr{Q}, \matr{R}} \matr{K}_k &= \begin{cases} \nabla_p \matr{K}_\infty \big\rvert_{p \in \matr{Q}, \matr{R}} &, \textrm{ if } k \geq N_i \\
                                           \nabla_p \matr{K}_k \big\rvert_{p \in \matr{Q}, \matr{R}}&, \textrm{ otherwise }
                                                \end{cases} \label{eq:lqr:K_tot_g} \\
    \nabla_{\theta^{\param{L}}} \pi_{\theta^{\param{L}}}^\param{L}(s_t) &= \nabla_{\matr{Q}, \matr{R}} \matr{K}_{t-i} (\hat{x}_t - \bar{x}_t)
\end{align}

\subsection{Summary of Control Algorithm with Learning}
The complete control algorithm is outlined in Algorithm \ref{alg:cs}. This shows the control algorithm in the exploration phase, but the exploitation phase is identical with the two exceptions: 1) there is no data collection (and therefore no parameter updates), and 2) we use the deterministic version (i.e. the mode of the probability distribution) of all policies except for the recomputation policy. The Bernoulli distribution of the recomputation policy does not generally tend to quasi-determinism as exploration settles, unlike the other policies. As such, we find that the deterministic version of this policy has worse control performance and less consistency in the plant response than the stochastic version which we are in fact optimizing the response for.

The system to be optimized runs in an episodic fashion, and every $Z$ steps the \gls{rl} algorithm updates the parameters of the control system. 

\begin{algorithm2e}[htbp]
    \SetAlgoLined
    \DontPrintSemicolon
    \caption{Control Algorithm with Learning}
    \While{Running}{
        Initialize episode: $x_0, \hat{p}_0$ \;
        Compute initial MPC solution: $u^{\param{M}}_{0:N_{\mathrm{max}} - 1}, \> \hat{x}_{1:N_{\mathrm{max}}} = \pi^{\param{M}}_{\theta^{\param{M}}}(x_0, \hat{p_0}, N_{\mathrm{max}})$ \;
        Execute first MPC input: $u_0^\param{M}$ \;
        Calculate LQR system matrices: $\matr{A}_{1:N_{\mathrm{max}}}, \matr{B}_{1: N_{\mathrm{max}}}$ \;
        \For{$t = 1, 2, \dots, T$}{
            Measure system state: $\bar{x}_{t}$ \;
            \eIf{$\pi^\param{c}_{\theta^\param{c}}(c_t \thinspace | \thinspace s_t)$ \textup{ draws } $\tilde{c}_t = 1$}{
                Compute MPC solution: $\tilde{u}^{\param{MN}}_{t: t + N_t - 1}, \> \hat{x}_{t+1: t+N_t} = \pi^\param{MN}_{\theta^{\param{MN}}}(\tilde{u}^{\param{MN}} | s_t)$ \;
                Execute control input: $\tilde{u}^{\param{MN}}_t$ \;
                                Calculate LQR system matrices: $\matr{A}_{t: t+N_t}, \matr{B}_{t:t+N_t}$ \;
                Update last computation states: $i \leftarrow t$ \;
            }{
                Compute input: $\tilde{u}^{\mathrm{CS}}_t = \pi^{\mathrm{CS}}_{\theta^{\param{M,L,N}}}(u^{\mathrm{CS}}_t | s_t)$ \;
                Execute input: $\tilde{u}^{\mathrm{CS}}_t$ \;
            }
            Collect data: $\mathcal{D} \leftarrow (s_t, \tilde{a}_t, R(s_t, \tilde{a}_t), s_{t+1})$ \;
            \If{$\mathrm{size}(\mathcal{D}) = Z$}{
                \For{$e = 1, \dots, \mathrm{Num Epochs}$}{
                    \For{$\mathcal{B} \in \mathcal{D}$}{
                        Evaluate \gls{ppo} objective over minibatch $\mathcal{B}$ \;
                        Update parameters \eqref{eq:rl:gd} \;
                    }
                }
                Empty dataset: $\mathcal{D} = \{\}$ \;
            }
        }
    }
    \label{alg:cs}
\end{algorithm2e}

\subsection{Reward Function}\label{sec:method:rewf}

The \gls{rl} reward function codifies the behaviour that the control algorithm is designed to exhibit, wrt. stability, control performance, and computational complexity. As \gls{rl} is a trial-and-error based optimization approach, the control algorithm will necessarily have to attempt risky maneuvers and experience the consequences in order to learn how to control the system. However, with a feasible composition of the learning problem in terms of hyperparameters and expressive power of the learned components, the end result of the converged behaviour should be stable and yield good control performance. A failure to achieve this would therefore indicate a misalignment between choice of reward function, and the control design requirements. With this in mind, we formulate a reward function on the following form:

\begin{align}
\begin{split}
     R(s_t, a_t) =& R_\ell(s_{t+1}) + \lambda_h(T - t) R_h(s_{t+1}) \> + \\
     &\lambda_c R_c(a_t) R_N(a_t) \label{eq:method:rewf}
\end{split}
\end{align}

Here, $\lambda_h$ and $\lambda_c$ are weighting factors that represents the relative importance of the different terms. $R_\ell$ is the control performance term which incentivizes the \gls{rl} policy to achieve good control performance. We set it to be the same as the stage cost from the \gls{mpc} objective, i.e. $R_\ell = -\ell$ but this is not a strict requirement. $R_h$ is a term that indicates whether any system constraints are violated. If using hard constraints as is the case in this paper, this term is binary and is further weighted by $T-t$, that is, the number of time steps remaining in the episode, since the episode is prematurely terminated if any constraints are violated. If the system to be controlled has soft constraints, the $R_h$ term could be made continuous. The $R_c$ term indicates whether the \gls{mpc} was computed at the current step ($R_c = 1$), and as such it should reflect the relative computational complexity of the of the two modes of the control system. Finally, $R_N$ represents the computational complexity of the \gls{mpc} as a function of the prediction horizon $N$. We assume that the computational complexity grows linearly in the prediction horizon, i.e. $R_N(N_t) = N_t$, as a lower bound for the true complexity. We favor a lower bound as this would bias the \gls{rl} policy towards better control performance rather than lower computation. The relationship between the prediction horizon and the computational complexity depends upon the algorithms used in the \gls{mpc} implementation. We employ an interior point method, which under the assumption of local convergence and a guess of an initial solution that is reasonable, generally yields linear complexity \citep{rao1998application}, while other methods such as active set methods typically yield quadratic growth in computational complexity \citep{lau_comparison_2015}. 

\subsection{Initialization of the Learning Procedure}
How to initialize the control algorithm, that is, defining its behaviour before any learning takes place is a question that has several valid answers. One could favor the most computationally expensive initialization, i.e. compute \gls{mpc} every step with the maximum horizon, which without any prior knowledge about the task would be the ``safest" initialization that is most likely to give the best control performance. This is however not necessarily the initialization that would facilitate the learning process most optimally, and might trap the learned components in local minima. To simplify the learning problem one might therefore instead favor initializing the components in the center of their operating range and with high entropy (wrt. to its output) such that exploration is maximal. 

We offered some guidance in the case of the \gls{lqr} in Section \ref{sec:theory:lqrcompat}, i.e. the \gls{lqr} should be initialized such that it is compatible with the \gls{mpc}. We view the purpose of the recomputation policy as finding instances where comparable (or better) control performance can be achieved without computing the \gls{mpc}, and as such, we choose to initialize it to emulate the \gls{mpc} paradigm and with high probability elect to compute the \gls{mpc}. For the horizon policy, we initialize it equal to the best performing fixed horizon \gls{mpc} scheme. Initializing the policies can be done by adjusting the bias terms of the policy outputs:

\begin{align}
    \pi_{\theta^\param{c}}^\param{c} &\leftarrow - \log\left(\frac{1}{c_{\mathrm{init}}} - 1\right) \label{eq:method:comp_init} \\
    \pi^\param{N}_{\theta^\param{N}} &\leftarrow \tanh^{-1}\left(\frac{(N_{\mathrm{max}} - N_{\mathrm{min}})(N_{\mathrm{init}} - (-1))}{1 - (-1)} + N_{\min}\right) \label{eq:method:horizon_init}
\end{align}

\section{Numerical Results}\label{sec:experiment}
This section illustrates the proposed control method as outlined in Section \ref{sec:method} on the simulated inverted pendulum system. Additionally, to aid in highlighting the contributions of the different meta-parameters of the control algorithm we also present experiments for each meta-parameter in isolation.

We set $N_\mathrm{min} = 1$ and $N_\mathrm{max}$ = 40, and note that this horizon does not cover a full swing-up maneuver of the pendulum (requiring $N>100$). The maximum horizon is chosen to emulate the effects of a computationally limited embedded hardware platform. The weighting terms of the reward function \eqref{eq:method:rewf} are set to $\lambda_h = -10$ and $\lambda_c = 10^{-2}$. This ensures that violating constraints incur a higher cost than finishing the episode, and makes the computation account for approximately 8\% of the total cost of the standard \gls{mpc} scheme (i.e. \gls{mpc} computed at every step) with the highest prediction horizon. 

While our method supports tuning parameters internal to the \gls{mpc} (i.e. objective, constraints etc.), doing so requires reevaluating the log-probabilities and recomputing the \gls{ocp} for every data sample after every parameter update (when using minibatches or several passes over a dataset). This adds considerable computation, and optimization of these parameters with \gls{rl} has been successfully demonstrated in previous works \cite{edwards2021automatic,gros_reinforcement_2021}. Therefore, since this paper focuses on optimizing the meta-parameters of the \gls{mpc} scheme (that is, the prediction horizon and when to compute) we do not tune any parameters internal to the \gls{mpc} in this example. We do however optimize the parameters of the \gls{lqr} to illustrate the concept, as the \gls{lqr} is less costly to evaluate.  Finally, we omit learning the value function of the \gls{mpc}, as it adds considerable complexity to the experiments, and was found to not add much benefit for control of similar systems in \cite{bohn2021reinforcement}.

\subsection{The Inverted Pendulum System} \label{sec:num:invp}
The inverted pendulum system is a classic control task in which a pendulum, consisting of a rigid rod with a mass at the end, is mounted on top of a cart that is fixed on a track. The control system exerts a horizontal force on the cart, which moves the cart back and forth on the track, which subsequently swings the pendulum. The control objective is to stabilize the pendulum in the up-position and positioning the cart at the position reference, while respecting the constraint that the cart's position is limited by the physical size of the track. The dynamics of the system are highly nonlinear, and further the system is unstable, meaning that a controller is necessary to guide the system to stable conditions and then to maintain the stability. We perturb the parameters of the \gls{mpc} dynamics model such that its dynamics differs from the plant dynamics, as shown in Table \ref{tab:invp:param}, and thus the \gls{lqr} is a useful addition to correct for prediction errors in between the \gls{mpc} computations. The state $x$ consists of the cart position $\psi$ and velocity $v$, pendulum angle $\phi$ and angular velocity $\omega$.


Each episode lasts a maximum of 150 steps (i.e. T=150), or until the position constraint \eqref{eq:invp:pos_h} is violated. We sample initial conditions according to \eqref{eq:invp:x0}, and a position reference \eqref{eq:invp:pos_r} that is redrawn every 50 steps, where $\mathcal{U}$ is the uniform distribution. The \gls{mpc} objective is defined as \eqref{eq:invp:ell}, i.e. minimize the kinetic energy of the system and the distance of the cart to the position reference, while maximizing the potential energy. 

\begin{align}
    \dot{v} &= \frac{m g \sin{\phi} \cos{\phi} - \frac{7}{3}\left(u + m l \omega^2 \sin{\phi} - \mu_c v\right) - \frac{\mu_p \omega \cos{\phi}}{l}}{m \cos^2{\phi} - \frac{7}{3}M} \\
    \dot{\omega} &= \frac{3}{7 l}\left(g \sin{\phi} - \dot{v} \cos{\phi} - \frac{u_p \omega}{m l} \right), \dot{\phi} = \omega, \enspace \dot{\psi} = v \\
    &x = \left[\psi, v, \phi, \omega\right]^\top, \enspace \psi_r(t) \sim \mathcal{U}(-1, 1) \label{eq:invp:pos_r} \\
    &x_0 \sim \left[0, \mathcal{U}(-1, 1), \mathcal{U}(-\pi, \pi), \mathcal{U}(-1, 1)\right]^\top, \matr{D} = \left[0.1\right]  \label{eq:invp:x0} \\
     &x^s_t = \left[\psi_r(t), 0, 0, 0\right]^\top, \enspace -5 \leq u_t \leq 5, \enspace -2 \leq \psi_t \leq 2 \label{eq:invp:pos_h} \\
     &\ell(x_t, u_t, \hat{p}_t) = \mathrm{E}_{\mathrm{k}} - 10 \mathrm{E}_{\mathrm{p}} + 10 (\psi_t - \psi_r(t))^2 \label{eq:invp:ell}
\end{align}

\begin{table} [htbp]
    \caption{Parameters of the inverted pendulum system.}
    \ra{1.2}
	\centering
	\begin{tabular}{@{}lrrl@{}} \toprule
		Name  & Plant & MPC & Description \\ \midrule
		$m$	&  0.1 & 0.2 & Mass of pendulum \\ 
		$M$ &  1.1 & 1.5 & Mass of cart and pendulum   \\
		$g$ & 9.81 & 9.81 & Gravitational constant \\	
		$l$ & 0.25 & 0.25 & Half the length of the pendulum \\
		$\mu_c$ & 0.01 & 0.01 & Friction coefficient between track and cart \\
		$\mu_p$ & 0.001 & 0.001 & Friction coefficient between pendulum and cart \\
		$\Delta_t$ & 0.04 & 0.04 & Discretization step size in seconds \\
		\bottomrule
	\end{tabular}
	\label{tab:invp:param} 
\end{table}

\subsection{Training and Evaluation}
The hyperparameters of the \gls{ppo} algorithm are listed in Appendix \ref{sec:appendixA}. The horizon and recomputation policies are fully-connected \glspl{nn} with 2 hidden layers with 64 nodes in each, whereas the value function is an \gls{nn} with 2 hidden layers of 128 nodes, all using the $\tanh$ activation function. We use the same hyperparameters for the isolated experiments, rather than tuning them specifically, and as such the isolation experiments serve mainly to assess how the two meta-parameters contribute. For the horizon experiment we use a frame-skip of $d=10$, and for the other experiments we use a varying $d \in [1, 2, 3, 4]$ that is drawn at the start of every episode (with $d = 1$ when evaluating).

We initialize the recomputation policy using \eqref{eq:method:comp_init} to compute the \gls{mpc} with 90\% probability, thus within two time steps there is a $99\%$ probability that the \gls{mpc} is computed. It therefore initially mimics the traditional \gls{mpc} scheme closely. The horizon policy is initialized at the best performing fixed-horizon (Fig. \ref{fig:exp:mpc_surface}), i.e. $N=31$, using equation \eqref{eq:method:horizon_init}.

To evaluate the performance of the control system governed by the \gls{rl} policy, we construct a "test-set" consisting of 25 episodes where all stochastic elements are drawn in advance (i.e. initial conditions and position references) such that the episodes are consistent across evaluations. This gives us an objective way to compare and order policies on their performance, and to compare against the \gls{mpc} baselines. We set the cost objective \eqref{eq:prob_definition} $C = -R$, and evaluate models based on the total sum of costs over the episodes in the test-set.

Moreover, for every experiment we report the average over five random initial seeds (referred to as models), which impact the initializations of \glspl{nn} and the randomness in exploration and episodes. Fig. \ref{fig:exp:mpc_surface} shows the total cost of the control algorithm baselines as a function of a static prediction horizon and recomputation schedule. The minimum cost is achieved with a prediction horizon of $N=31$ and a schedule of recompute at every step. It is important to note that while this figure indicates that the optimization landscape as a function of these two variables is monotonic and amenable to optimization, the reward landscape as a function of the parameters $\theta$ (which in this example consists of the parameters of the \gls{lqr} and the parameters of the recomputation and horizon \glspl{nn}) is likely very different --- containing many valleys and hills to overcome in order to minimize the cost objective.

\begin{figure}
    \centering
    \includegraphics[width=0.5\textwidth]{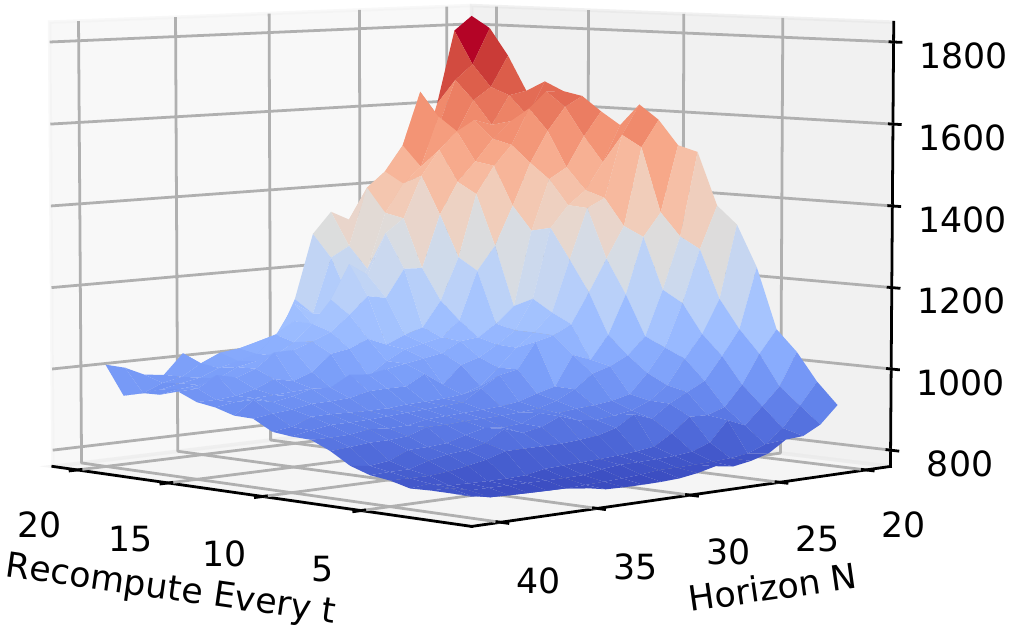}
    \caption{Total cost over the evaluation set for the MPC as a function of fixed horizons and fixed recomputation schedules. The minimum is found at horizon $N=31$ and recompute every step with a cost of $781$.}
    \label{fig:exp:mpc_surface}
\end{figure}

\subsection{Inverted Pendulum Results and Discussion}
The results of the experiments are presented in Figures \ref{fig:exp:training} and \ref{fig:exp:distributions}, and Table \ref{tab:exp:ablation}. 
When reporting the cost of the tuned system, we average over the best performing policy of each model, as well as over five seeds to account for the stochastic recomputation policy. When describing learned behaviour as in Fig. \ref{fig:exp:distributions} however, we use the best performing policy and one specific seed. We summarize our main findings as follows:

\textbf{Learning improves both computational complexity and control performance}. The \gls{rl} framework we propose is able to improve upon the total objective by $21.5\%$, which interestingly is not only due to reductions in the computation term (70.7\%) but also a sizable improvement on the control performance objective ($18.4\%$). For this example it takes about 20 hours of data of interaction with the system to reach convergence, corresponding to 700 thousand data samples. Fig. \ref{fig:exp:training:full} shows that generally the performance  monotonically improves with more data, we therefore do not view the data requirement as a major issue. Note that we did not spend significant time tuning the hyperparameters of the \gls{rl} algorithm, and we favored consistency over faster rate of improvement. Thus, the algorithm's data requirement could potentially be improved with hyperparameter optimization. Fig. \ref{fig:exp:training} shows that the recomputation meta-parameter is the most impactful, reaching a cost of around 700 when optimized by itself, which is about 13\% higher than the converged value of the optimized complete policy. It also shows why we favor biasing towards higher computation, as while initializing the recomputation policy to compute the \gls{mpc} with 10\% probability reaches the same asymptotic performance as the 90\% initialization, it is the only policy we trained that intermittently violates the constraint.

Because we initialize the learning problem at a best-effort of optimal tuning, the control performance improvements we observe are non-trivial to explain and arise due to complex interactions between multi-step adaptive-horizon \gls{mpc} and the \gls{lqr} control law. Fig. \ref{fig:exp:distributions} shows the distribution of the prediction horizon, and the steps between \gls{mpc} computations chosen by the policies. The horizon policy has converged to only selecting the maximum horizon, while the \gls{mpc} is mostly computed at every step ($70\%$ of steps) with some significantly longer streaks. Because of the finite-horizon nature of the \gls{ocp} the \gls{mpc} will produce solutions of varying optimality based on the exact initial conditions it is computed from. The \gls{rl} policy has learned to recognize a set of conditions for which the computed solutions are more optimal than neighboring conditions, and therefore not recomputing and employing a longer section of the more optimal input sequence will produce better control performance.


The control algorithm we propose adds some overhead compared to \gls{mpc}, i.e. evaluating the policies deciding if the \gls{mpc} should be computed and with what horizon, and computing the \gls{lqr} gain-matrix as necessary. The overhead is however small compared to the execution time of the \gls{mpc}, and therefore our framework results in a 36\% reduction in the total processing time of the control system compared to the best performing \gls{mpc} scheme. This frees up resources for other on-board tasks the controlled system might have, or could be leveraged to increase the battery life of the system. 

One of the models we trained for the complete policy, and one of the models for the isolated recomputation policy got stuck in the local minima of computing the \gls{mpc} at every step, and since this didn't give any interesting results (essentially maintaining the initial behaviour and performance) we excluded them from Fig. \ref{fig:exp:training} and the discussions, replacing them with new models with new seeds. Since policy gradient algorithms are local search methods it is to be expected that it finds local minima, and random exploration can cause it to sometimes settle in sub-optimal local minima. This is also a question of how much data is used to generate the gradient.

\begin{figure}[htbp]
    \subfloat{\includegraphics[width=0.25\textwidth]{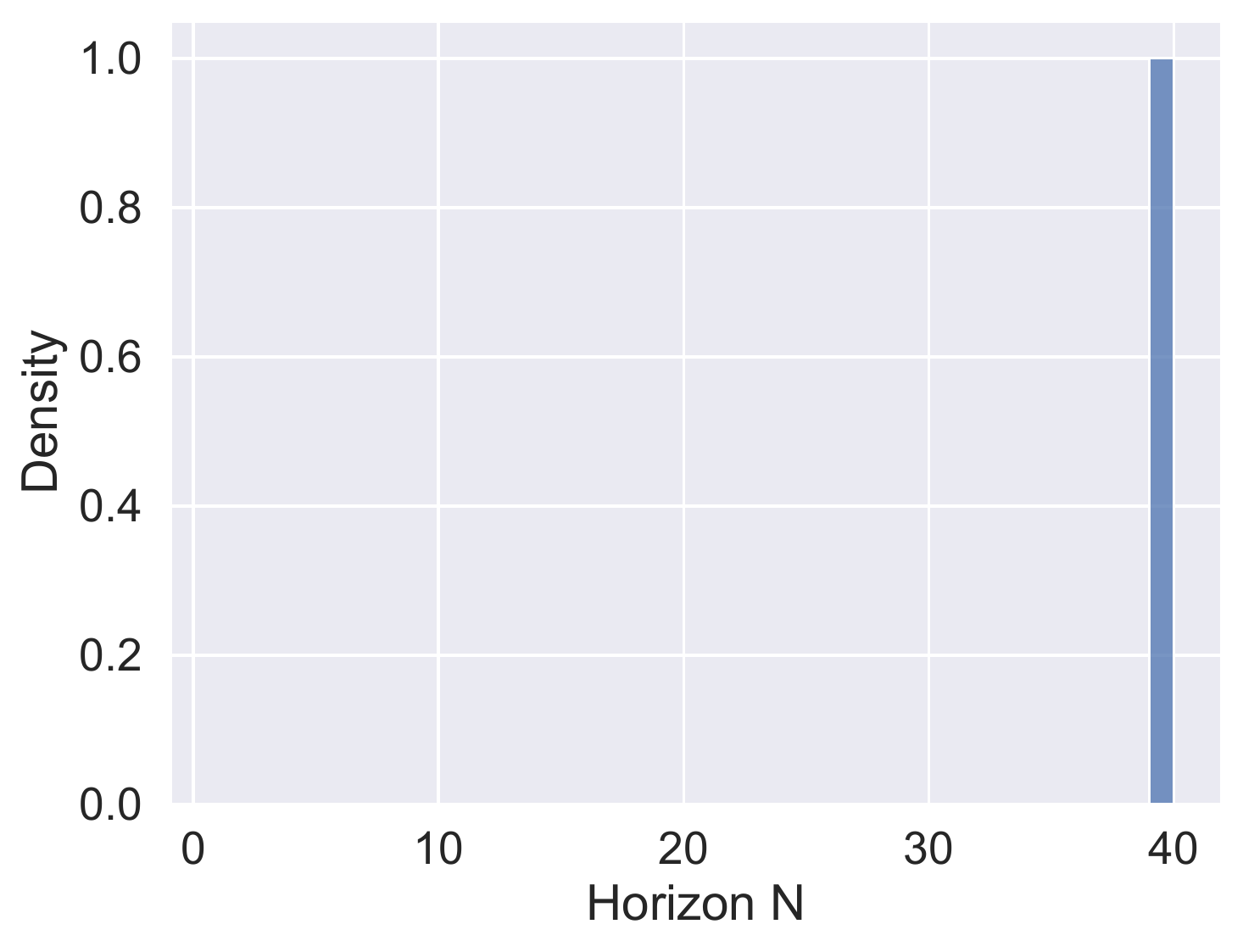}}~
    \subfloat{\includegraphics[width=0.25\textwidth]{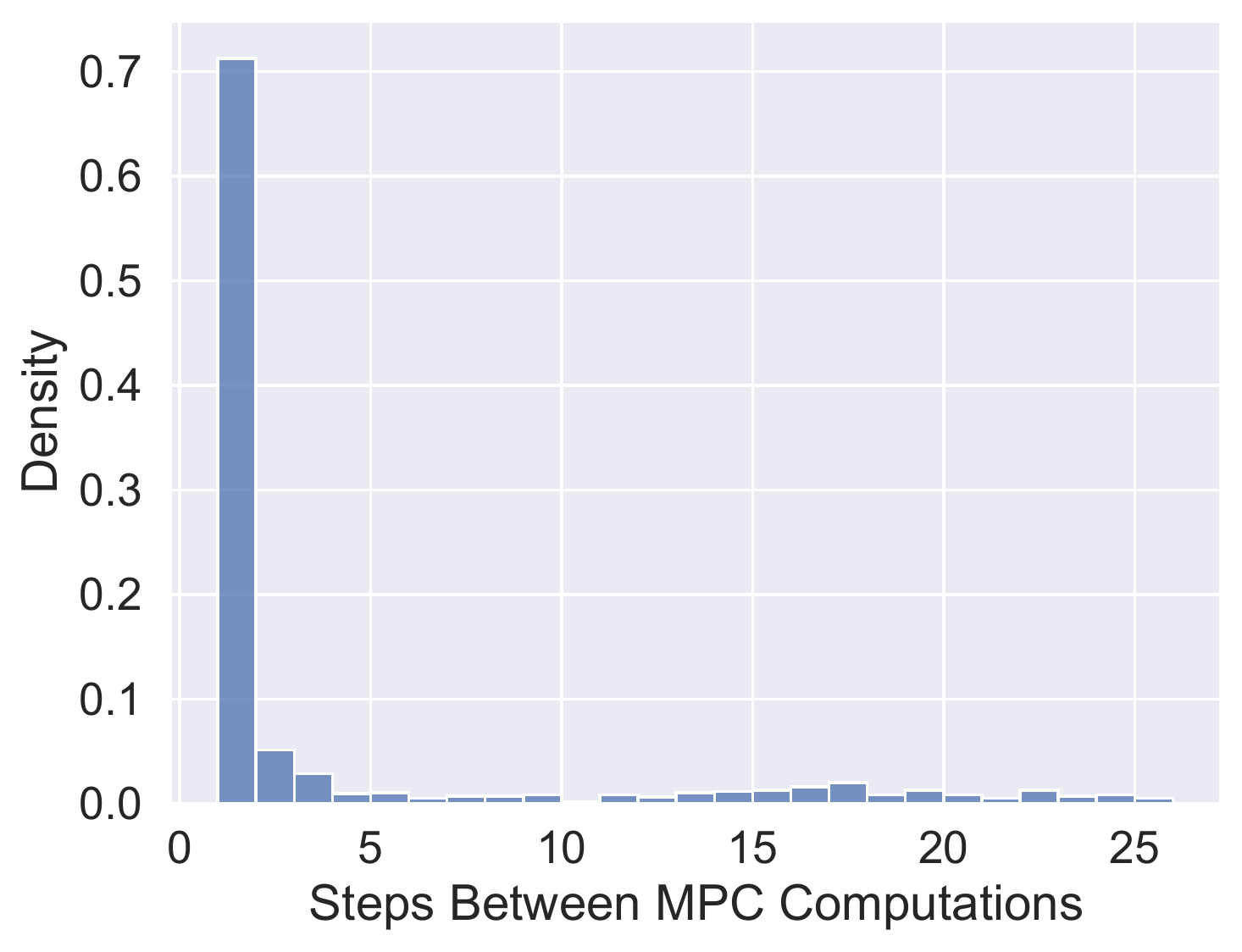}}
    \caption{Distribution of prediction horizons and steps between \gls{mpc} computations selected by the best performing policy on the evaluation test-set.}
    \label{fig:exp:distributions}
\end{figure}

\textbf{Joint optimization delivers additional improvements}. These results support the conjecture stated in the introduction; by optimizing the different parameters of the control algorithm together, we are able to enhance the control algorithm in ways that we are not able to when optimizing the same parameters in isolation. The horizon policy tends to a mean horizon of 28-31 when optimized by itself (Fig. \ref{fig:exp:training:pi_N}), which is consistent with the best performing fixed horizon \gls{mpc} as seen in Fig. \ref{fig:exp:mpc_surface}. However, as Fig. \ref{fig:exp:distributions} shows the complete \gls{rl} policy favored higher horizons when optimizing all meta-parameters together. As shown in Table \ref{tab:exp:ablation} this improves performance significantly, where imposing a maximum horizon of 31 results in a $13.4\%$ increase in cost. When solely tuning the \gls{lqr} applied on an \gls{mpc} computed on a fixed recomputation schedule, we were not able to achieve any consistent improvements. We hypothesize that this is due to the model mismatch (the \gls{mpc} is overestimating the weight of the system by $36\%$) such that the computed \gls{lqr} is not very well suited for the actual control problem, and therefore its gradients are flat and noisy. When combining \gls{lqr} tuning with meta-parameter optimization we are able to tune the \gls{lqr} to achieve meaningful improvements, with the tuned \gls{lqr} improving by $1.95\%$ over the initialization described in Section \ref{sec:theory:lqrcompat}. This might be due the recomputation policy generating trajectories for the \gls{lqr} that are similar, thus yielding more consistent gradients. The last entry in Table \ref{tab:exp:ablation} shows a scenario where the same amount of computation is expended uniformly in time rather than dynamically allocated by the recomputation policy, resulting in a $35.45\%$ increase in cost.



\begin{figure*}
    \subfloat[Optimizing both meta-parameters jointly. The MPC baseline represents the \gls{mpc} scheme tuned as a function of fixed recomputation schedule and horizon.]{\includegraphics[width=\textwidth]{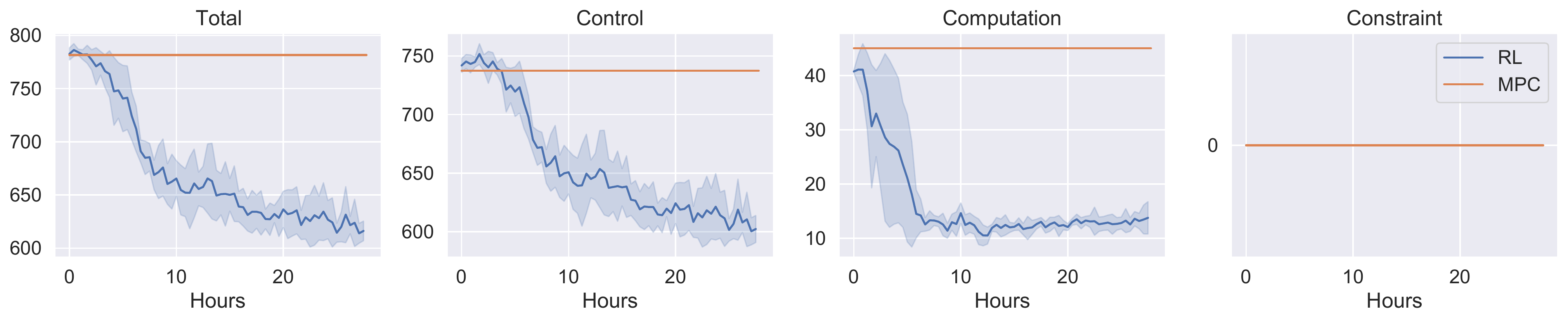}\label{fig:exp:training:full}}
    
    
    \subfloat[Optimizing the recomputation meta-parameter in isolation.]{\includegraphics[width=\textwidth]{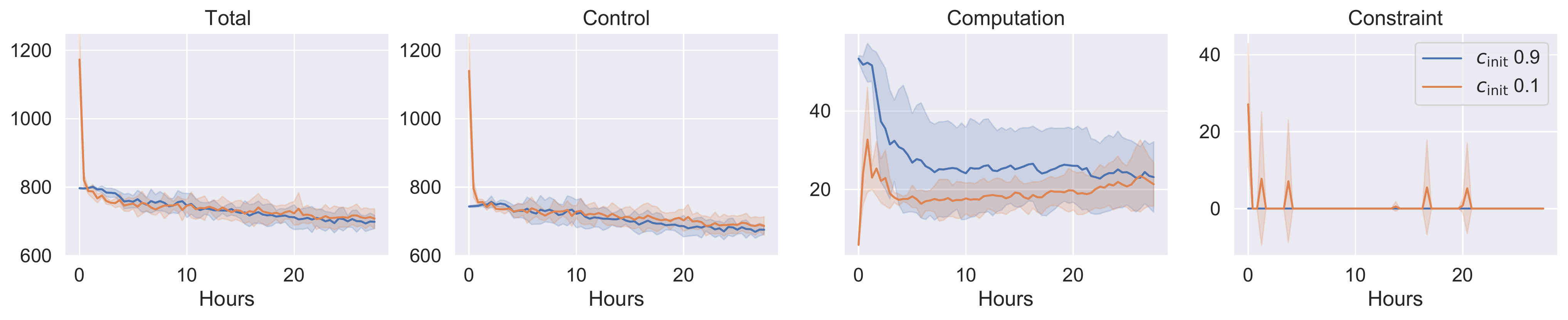} \label{fig:exp:training:pi_c}}
    
    \subfloat[Optimizing the prediction horizon meta-parameter in isolation.]{\includegraphics[width=\textwidth]{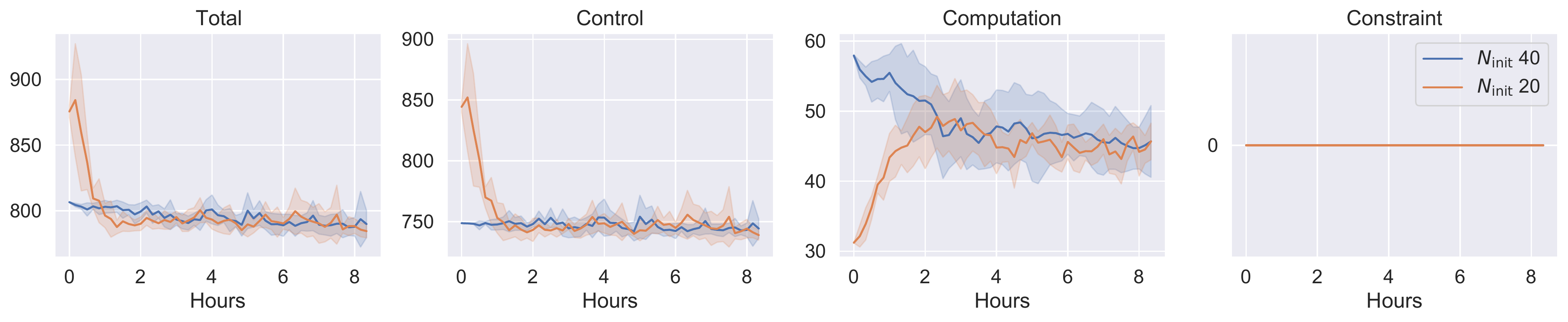} \label{fig:exp:training:pi_N}}
    \caption{Training process for the proposed learned control algorithm, and for each meta-parameter in isolation. In the isolated cases, we show that the tuning process is capable of recovering from sub-optimal initializations.}
    \label{fig:exp:training}
\end{figure*}

\begin{table} [htbp]
    \caption{Ablation analysis: We take the best performing policy and alter one aspect at a time, observing its effect on the cost.}
    \ra{1.2}
	\centering
	\begin{tabular}{@{}lr@{}} \toprule
		Scenario  & Change \\ \midrule
		Default \gls{lqr} tuning & $+1.95\%$ \\
		Max horizon 31 & $+13.41\%$\\
		Recompute at the average frequency (every $t=4$) & $+35.45\%$\\
		\bottomrule
	\end{tabular}
	\label{tab:exp:ablation} 
\end{table}

\section{Conclusion}
This paper has introduced a novel control algorithm that is tuned with \gls{rl} to jointly improve the computational power usage and the control performance. We focus on optimizing the meta-parameters of the \gls{mpc} scheme, and demonstrate its efficacy on the classic control task of balancing an inverted pendulum. We show that by selecting the conditions under which the \gls{mpc} is computed, control performance can be improved over the paradigm of computing at every step, and that the control algorithm can be further improved by considering all the optimized parameters together. Seeing as \gls{mpc} is increasingly being considered for applications with fast dynamics or limited computational power and energy resources, our framework could be an important tool in enabling such applications to harness the good control performance and constraint satisfaction abilities of the \gls{mpc}. We found that with model mismatch, tuning the dual mode \gls{mpc} and \gls{lqr} control law was difficult. Future work could evaluate if state-dependent $\matr{Q}$ and $\matr{R}$ matrices alleviates this issue, which would also provide a gain-scheduling property to the control algorithm. Whether any stability guarantees and control satisfaction properties can be given under the learned meta-parameter-deciding policies should be investigated.  

\bibliographystyle{ieeetr}
\bibliography{main}

\begin{thebibliography}{10}

\bibitem{gros2012aircraft}
S.~Gros, R.~Quirynen, and M.~Diehl, ``Aircraft control based on fast non-linear
  mpc \& multiple-shooting,'' in {\em 2012 IEEE 51st IEEE Conference on
  Decision and Control (CDC)}, pp.~1142--1147, IEEE, 2012.

\bibitem{albin2015nonlinear}
T.~Albin, D.~Ritter, D.~Abel, N.~Liberda, R.~Quirynen, and M.~Diehl,
  ``Nonlinear mpc for a two-stage turbocharged gasoline engine airpath,'' in
  {\em 2015 54th IEEE Conference on Decision and Control (CDC)}, pp.~849--856,
  IEEE, 2015.

\bibitem{Joh17}
T.~A. Johansen, ``Toward dependable embedded model predictive control,'' {\em
  IEEE Systems Journal}, vol.~11, pp.~1208--1219, 2017.

\bibitem{le_approximate_2012}
L.~Feng, C.~Gutvik, T.~Johansen, D.~Sui, and A.~Brubakk, ``Approximate explicit
  nonlinear receding horizon control for decompression of divers,'' {\em
  Control Systems Technology, IEEE Transactions on}, vol.~20, pp.~1275--1284,
  09 2012.

\bibitem{gondhalekar_tackling_2015}
R.~Gondhalekar, E.~Dassau, and F.~J. {Doyle III}, ``Tackling problem
  nonlinearities \& delays via asymmetric, state-dependent objective costs in
  mpc of an artificial pancreas,'' {\em IFAC-PapersOnLine}, vol.~48, no.~23,
  pp.~154 -- 159, 2015.
\newblock 5th IFAC Conference on Nonlinear Model Predictive Control NMPC 2015.

\bibitem{scokaert_suboptimal_1999}
P.~O.~M. {Scokaert}, D.~Q. {Mayne}, and J.~B. {Rawlings}, ``Suboptimal model
  predictive control (feasibility implies stability),'' {\em IEEE Transactions
  on Automatic Control}, vol.~44, no.~3, pp.~648--654, 1999.

\bibitem{goebel_simple_2015}
G.~Goebel and F.~Allgöwer, ``A simple semi-explicit mpc algorithm,'' {\em
  IFAC-PapersOnLine}, vol.~48, no.~23, pp.~489 -- 494, 2015.
\newblock 5th IFAC Conference on Nonlinear Model Predictive Control NMPC 2015.

\bibitem{bemporad_explicit_2002}
A.~Bemporad, M.~Morari, V.~Dua, and E.~N. Pistikopoulos, ``The explicit linear
  quadratic regulator for constrained systems,'' {\em Automatica}, vol.~38,
  no.~1, pp.~3 -- 20, 2002.

\bibitem{michalska_robust_1993}
H.~{Michalska} and D.~Q. {Mayne}, ``Robust receding horizon control of
  constrained nonlinear systems,'' {\em IEEE Transactions on Automatic
  Control}, vol.~38, no.~11, pp.~1623--1633, 1993.

\bibitem{krener_adaptive_2018}
A.~J. Krener, ``Adaptive {Horizon} {Model} {Predictive} {Control},'' {\em
  IFAC-PapersOnLine}, vol.~51, pp.~31--36, Jan. 2018.

\bibitem{scokaert_min_1998}
P.~O.~M. {Scokaert} and D.~Q. {Mayne}, ``Min-max feedback model predictive
  control for constrained linear systems,'' {\em IEEE Transactions on Automatic
  Control}, vol.~43, no.~8, pp.~1136--1142, 1998.

\bibitem{gardezi_machine_2018}
M.~S.~M. {Gardezi} and A.~{Hasan}, ``Machine learning based adaptive prediction
  horizon in finite control set model predictive control,'' {\em IEEE Access},
  vol.~6, pp.~32392--32400, 2018.

\bibitem{bohn2021reinforcement}
E.~Bøhn, S.~Gros, S.~Moe, and T.~A. Johansen, ``Reinforcement learning of the
  prediction horizon in model predictive control,'' {\em IFAC-PapersOnLine},
  vol.~54, no.~6, pp.~314--320, 2021.
\newblock 7th IFAC Conference on Nonlinear Model Predictive Control NMPC 2021.

\bibitem{gros2019data}
S.~Gros and M.~Zanon, ``Data-driven economic nmpc using reinforcement
  learning,'' {\em IEEE Transactions on Automatic Control}, vol.~65, no.~2,
  pp.~636--648, 2019.

\bibitem{edwards2021automatic}
W.~Edwards, G.~Tang, G.~Mamakoukas, T.~Murphey, and K.~Hauser, ``Automatic
  tuning for data-driven model predictive control,'' in {\em International
  Conference on Robotics and Automation (ICRA)}, 2021.

\bibitem{bohn2021optimization}
E.~Bøhn, S.~Gros, S.~Moe, and T.~A. Johansen, ``Optimization of the model
  predictive control update interval using reinforcement learning,'' {\em
  IFAC-PapersOnLine}, vol.~54, no.~14, pp.~257--262, 2021.
\newblock 3rd IFAC Conference on Modelling, Identification and Control of
  Nonlinear Systems MICNON 2021.

\bibitem{allgower_nonlinear_1999}
F.~Allg{\"o}wer, T.~A. Badgwell, J.~S. Qin, J.~B. Rawlings, and S.~J. Wright,
  ``Nonlinear predictive control and moving horizon estimation --- an
  introductory overview,'' in {\em Advances in Control} (P.~M. Frank, ed.),
  (London), pp.~391--449, Springer London, 1999.

\bibitem{rawlings2017model}
J.~B. Rawlings, D.~Q. Mayne, and M.~Diehl, {\em Model predictive control:
  theory, computation, and design}, vol.~2.
\newblock Nob Hill Publishing Madison, WI, 2017.

\bibitem{lowrey2018plan}
K.~Lowrey, A.~Rajeswaran, S.~Kakade, E.~Todorov, and I.~Mordatch, ``Plan
  online, learn offline: Efficient learning and exploration via model-based
  control,'' {\em arXiv preprint arXiv:1811.01848}, 2018.

\bibitem{zhong_value_2013}
M.~{Zhong}, M.~{Johnson}, Y.~{Tassa}, T.~{Erez}, and E.~{Todorov}, ``Value
  function approximation and model predictive control,'' in {\em 2013 IEEE
  Symposium on Adaptive Dynamic Programming and Reinforcement Learning
  (ADPRL)}, pp.~100--107, 2013.

\bibitem{esfahani2021reinforcement}
H.~Nejatbakhsh~Esfahani, A.~Bahari~Kordabad, and S.~Gros, ``Reinforcement
  learning based on mpc/mhe for unmodeled and partially observable dynamics,''
  pp.~2121--2126, 05 2021.

\bibitem{bemporad1999robust}
A.~Bemporad and M.~Morari, ``Robust model predictive control: A survey,'' in
  {\em Robustness in identification and control}, pp.~207--226, Springer, 1999.

\bibitem{mayne2005robust}
D.~Q. Mayne, M.~M. Seron, and S.~Rakovi{\'c}, ``Robust model predictive control
  of constrained linear systems with bounded disturbances,'' {\em Automatica},
  vol.~41, no.~2, pp.~219--224, 2005.

\bibitem{grune}
L.~Gr\"{u}ne and J.~Pannek, {\em Nonlinear Model Predictive Control}.
\newblock Springer-Verlag, London, 2011.

\bibitem{mayne2000constrained}
D.~Mayne, J.~Rawlings, C.~Rao, and P.~Scokaert, ``Constrained model predictive
  control: Stability and optimality,'' {\em Automatica}, vol.~36, no.~6,
  pp.~789--814, 2000.

\bibitem{bertsekas1995dynamic}
D.~P. Bertsekas, {\em Dynamic programming and optimal control}, vol.~1.
\newblock Athena scientific Belmont, MA, 1995.

\bibitem{sutton_reinforcement_2018}
R.~S. Sutton and A.~G. Barto, {\em Reinforcement Learning: An Introduction}.
\newblock Cambridge, MA, USA: A Bradford Book, 2018.

\bibitem{schulman2017proximal}
J.~Schulman, F.~Wolski, P.~Dhariwal, A.~Radford, and O.~Klimov, ``Proximal
  policy optimization algorithms,'' {\em arXiv preprint arXiv:1707.06347},
  2017.

\bibitem{wang1997modeling}
W.~Wang and F.~Famoye, ``Modeling household fertility decisions with
  generalized poisson regression,'' {\em Journal of population economics},
  vol.~10, no.~3, pp.~273--283, 1997.

\bibitem{demirtasAccurate2017}
H.~Demirtas, ``On accurate and precise generation of generalized poisson
  variates,'' {\em Communications in Statistics - Simulation and Computation},
  vol.~46, no.~1, pp.~489--499, 2017.

\bibitem{mnih2015human}
V.~Mnih, K.~Kavukcuoglu, D.~Silver, A.~A. Rusu, J.~Veness, M.~G. Bellemare,
  A.~Graves, M.~Riedmiller, A.~K. Fidjeland, G.~Ostrovski, {\em et~al.},
  ``Human-level control through deep reinforcement learning,'' {\em nature},
  vol.~518, no.~7540, pp.~529--533, 2015.

\bibitem{kalyanakrishnan2021analysis}
S.~Kalyanakrishnan, S.~Aravindan, V.~Bagdawat, V.~Bhatt, H.~Goka, A.~Gupta,
  K.~Krishna, and V.~Piratla, ``An analysis of frame-skipping in reinforcement
  learning,'' {\em arXiv preprint arXiv:2102.03718}, 2021.

\bibitem{friedman2001elements}
J.~Friedman, T.~Hastie, R.~Tibshirani, {\em et~al.}, {\em The elements of
  statistical learning}, vol.~1.
\newblock Springer series in statistics New York, 2001.

\bibitem{gros_reinforcement_2021}
S.~Gros and M.~Zanon, ``Reinforcement learning based on mpc and the stochastic
  policy gradient method,'' in {\em 2021 American Control Conference (ACC)},
  pp.~1947--1952, 2021.

\bibitem{laub1979schur}
A.~Laub, ``A schur method for solving algebraic riccati equations,'' {\em IEEE
  Transactions on automatic control}, vol.~24, no.~6, pp.~913--921, 1979.

\bibitem{rao1998application}
C.~V. Rao, S.~J. Wright, and J.~B. Rawlings, ``Application of interior-point
  methods to model predictive control,'' {\em Journal of optimization theory
  and applications}, vol.~99, no.~3, pp.~723--757, 1998.

\bibitem{lau_comparison_2015}
M.~Lau, S.~Yue, K.~Ling, and J.~Maciejowski, ``A comparison of interior point
  and active set methods for fpga implementation of model predictive control,''
  {\em Proc. European Control Conference}, 03 2015.

\end{thebibliography}

\newpage
\begin{appendices}
\section{PPO Hyperparameters}\label{sec:appendixA}
\begin{table} [hbt!]
    \caption{Hyperparameters of the PPO algorithm}
    \ra{1.2}
	\centering
	\begin{tabular}{@{}lrl@{}} \toprule
		Name  & Value & Description \\ \midrule
		$\gamma$ & $0.99$ & Discount factor \\
		Z & $256$ & Length of trajectory collected by each actor \\
		n\_envs & $4$ & Number of actors run in parallel \\
		ent\_coef & $0$ & Coefficient for entropy loss term \\
		$\eta$ & $3 \cdot 10^{-4}$ & Learning rate \\
		vf\_coef & $0.5$ & Coefficient for value function loss term \\
		max\_grad\_norm & $0.5$ & Maximum global norm of gradients \\
		$\lambda$ & 0.9 & Bias-variance tradeoff for GAE \\
		nminibatches & 1 & Number of partitions for dataset \\
		noptepochs & 10 & Number of passes over the dataset \\
		$\epsilon$ & 0.25 & Clip parameter for the objective function \\
		\bottomrule
	\end{tabular}
	\label{tab:appA:hyperparams} 
\end{table}
\end{appendices}

\end{document}